%% file: main.tex
\documentclass[journal]{IEEEtran}  
\IEEEoverridecommandlockouts
\usepackage{amsmath,amssymb}
\usepackage{graphicx}
\usepackage{bm}
\usepackage{svg}
\usepackage{textcomp}
\usepackage{xcolor}
\usepackage{cite}
\usepackage{url}
\usepackage[hidelinks]{hyperref}
\usepackage{algorithm}
\usepackage[noend]{algorithmic}
\usepackage{xprintlen}
\usepackage{cases}
\usepackage{tabularx,ragged2e}
\usepackage{booktabs}
\usepackage{nccmath}
\usepackage{caption}
\usepackage{float}
\usepackage[labelformat=simple]{subcaption}
\usepackage{mathtools}
\usepackage{verbatim}
\usepackage{import}

\def\BibTeX{{\rm B\kern-.05em{\sc i\kern-.025em b}\kern-.08em
    T\kern-.1667em\lower.7ex\hbox{E}\kern-.125emX}}
\usepackage{array}
\newcolumntype{P}[1]{>{\centering\arraybackslash}p{#1}}

\input{settings}
\begin{document}
\title{Control of Multi-agent Systems under STL Specifications based on Prescribed Performance Observers}
\author{Tommaso Zaccherini, \IEEEmembership{Student Member, IEEE}, Siyuan Liu, \IEEEmembership{Member, IEEE}, \\ Dimos V. Dimarogonas, \IEEEmembership{Fellow, IEEE}
\thanks{This work was supported in part by the Wallenberg AI, Autonomous Systems and
Software Program (WASP) funded by the Knut and Alice Wallenberg (KAW) Foundation, Horizon Europe EIC project SymAware (101070802), 
the ERC LEAFHOUND Project, and the Swedish Research Council (VR).
}
\thanks{Tommaso Zaccherini and Dimos V. Dimarogonas are with the Division
of Decision and Control Systems, KTH Royal Institute of Technology, Stockholm, Sweden.
	E-mail: {\tt\small \{tommasoz, dimos\}@kth.se}. Siyuan Liu is with the Department of Electrical Engineering, Control Systems Group, Eindhoven University of Technology, the Netherlands. Email: {\tt\small s.liu5@tue.nl}}
}

\maketitle

\input{Content/Abstract}

\begin{IEEEkeywords}
Multi-agent systems, Signal Temporal Logic, Observer-based control, Task/communication graphs mismatch, Decentralized control and estimation.
\end{IEEEkeywords}

\section{Introduction}
\label{sec:introduction}
\input{Content/Introduction}

\section{Preliminaries and Problem Setting}\label{Sec: Preliminaries}
\input{Content/Preliminaries}

\section{Decentralized State Observer}\label{Sec: Decentralized State Observer}

\input{Content/Observer}

\section{Cluster-induced graph and Observer-based Task Satisfaction}\label{Sec: Observer-based Task Satisfaction}
\input{Content/Observer_based_Task_Satisfaction}

\section{Distributed Controller Design}
\label{Sec: Distributed Controller Design}
\input{Content/Controller}

\section{Case Study}
\label{Sec: Case Study}
\input{Content/Case_study}

\section{Conclusions}
\label{Sec: Conclusions}
\input{Content/Conclusions}

\bibliographystyle{IEEEtran}
\bibliography{reference}

\end{document}

%% file: Settings.tex
\newtheorem{problem}{\textbf{Problem}}
\newtheorem{proposition}{\bf{Proposition}}
\newtheorem{definition}{\bf{Definition}}
\newtheorem{lemma}{\bf{Lemma}}
\newtheorem{theorem}{\bf{Theorem}}
\newenvironment{proof}
{\par\noindent\textit{Proof:}\ }
{\hfill$\blacksquare$\par}
\newtheorem{remark}{\bf{Remark}}
\newtheorem{assumption}{\bf{Assumption}}

\newtheorem{corollary}{\bf{Corollary}}
\newtheorem{example}{\textbf{Example}}

\def\myunderbar#1{\underline{\sbox\tw@{$#1$}\dp\tw@\z@\box\tw@}}

\newcommand{\vect}[1]{\boldsymbol{#1}}
\newcommand{\norm}[1]{ \lVert #1 \rVert}
\newcommand{\abs}[1]{ | #1 |}

\newcommand{\Nhop}[2]{\mathcal{N}_{#1}^{{#2}\text{-hop}}} 
\newcommand{\khop}[1]{\text{{#1}c}}
\newcommand{\Neigh}[2]{\mathcal{N}_{#1}^{{#2}}} 
\newcommand{\Neighext}[2]{\Bar{\mathcal{N}}_{#1}^{{#2}}} 
\newcommand{\estimate}[3]{\hat{{#1}}^{{#2}}_{{#3}}}
\newcommand{\state}[3]{{{#1}}^{{#2}}_{{#3}}}
\newcommand{\error}[3]{\Tilde{{#1}}^{{#2}}_{{#3}}}
\newcommand{\statediff}[5]{{#1}^{#2}_{#3}-{#4}^{#5}_{#3}}

\newcommand{\kc}[2]{{#1}_{#2}^{\text{kc}}}

%% file: Content/Abstract.tex
\begin{abstract}

This paper addresses decentralized control of \mbox{large-scale} heterogeneous \mbox{multi-agent} systems subject to bounded external disturbances and limited communication, with the objective of satisfying cooperative Signal Temporal Logic (STL) specifications. 
The considered specifications involve spatiotemporal tasks that require collaboration among multiple agents, including agents beyond direct communication neighborhoods.
To address the communication constraints, a \mbox{$k$-hop} Prescribed Performance State Observer (\mbox{$k$-hop} PPSO) is designed to enable each agent to estimate the states of agents up to $k$ communication hops away using only information from $1$-hop neighbors, while guaranteeing predefined performance bounds on the estimation errors. The estimation error bounds are explicitly incorporated into a reformulation of the spatial robustness of the STL specifications, yielding robustness measures that account for \mbox{worst-case} estimation uncertainty. Based on the modified robustness, a decentralized \mbox{continuous-time} feedback control law is designed to guarantee satisfaction of the STL specifications in the presence of bounded disturbances and estimation errors. The proposed framework provides formal correctness guarantees using only local information and limited communication. Numerical simulations illustrate the theoretical results.
\end{abstract}

%% file: Content/Introduction.tex
\IEEEPARstart{T}{he} growing deployment of robotic platforms in complex and \mbox{large-scale} environments has shifted control system design from centralized architectures toward distributed and cooperative frameworks. Applications such as \mbox{multi-robot} exploration \cite{1435481}, coordinated unmanned aerial vehicle formations \cite{1068000}, autonomous vehicle platooning \cite{7497531}, and distributed manipulation \cite{7353473} require multiple dynamical systems to operate concurrently while interacting through physical couplings, shared objectives, or communication networks. In such settings, centralized control strategies often suffer from scalability limitations, high communication overhead, and reduced robustness to failures. Consequently, \mbox{multi-agent} systems (MAS) have emerged as a natural framework for modeling, analysis, and control of interconnected robotic agents capable of autonomous \mbox{decision-making}.

In recent years, control methodologies for MAS have been extended to accommodate complex task specifications expressed in temporal logic \cite{maler2004monitoring}. Temporal logic—such as Linear Temporal Logic (LTL) and Signal Temporal Logic (STL)—shares structural similarities with natural language and provides a robust formalism for precisely analyzing system behaviors over time. As a result, it enables the specification of sophisticated, high-level requirements that go beyond conventional control objectives \cite{belta2017formal}.

Planning and control under temporal logic specifications have been extensively studied for both single-agent (SAS) and multi-agent systems (MAS). Planning under LTL specifications is typically addressed by constructing a product transition system that captures both the system dynamics and the specification automaton \cite{loizou2004automatic, liu2017distributed, guo2015multi, 10383968}. In contrast, planning under STL specifications is often formulated as a mixed-integer linear programming (MILP) problem \cite{sun2022multi, 7039363} or approached through \mbox{sampling-based} methods as  Rapidly Exploring Random Trees (RRT) and Rapidly Exploring Random Trees Star (RRT*) \cite{vasile2017sampling, karlsson2020sampling, 10341993}. Despite their effectiveness, these approaches face significant scalability challenges. In particular, while \mbox{LTL-based} methods are affected by the big dimensional state-space resulting from the construction of product automata, \mbox{STL-based} MILP formulations and sampling-based methods exhibit a combinatorial growth in the number of binary decision variables and possible trajectory combinations as the number of agents and tasks increases.

To mitigate these limitations, recent works exploit the notion of space robustness \cite{donze2010robust}—naturally associated with STL—to design controllers capable of driving SAS and MAS toward the satisfaction of STL tasks. 
For instance, \cite{lindemann2019control, 9246252,10885877, 9029666, 9483028} exploit Control Barrier Functions (CBFs) \cite{ames2019control} to enforce robustness constraints, while \cite{LINDEMANN2021100973, 10918825, long2024dynamic} leverage the notion of Prescribed Performance Control (PPC) \cite{4639441}.
However, these approaches typically assume the availability of communication links among cooperating agents, which limits their applicability in scenarios where agents must collaborate without direct communication. To address this issue, \cite{11222720} proposes a task decomposition algorithm that redefines tasks involving \mbox{non-communicating} agents to align them with the available communication topology. Nevertheless, this approach requires solving an optimization problem to decompose the task along a selected path, introducing an additional layer of complexity to the control problem and creating a dependency on the chosen path.

Motivated by these limitations and by the growing interest in distributed state estimation for MAS \cite{ACIKMESE20141037,6308694, 8264361, ZHAO201322, XU2013334, LI2011510, 7440802}, this paper proposes a decentralized \mbox{observer-based} control framework that enables MAS to satisfy global STL specifications despite the absence of direct communication among collaborating agents.

We consider a MAS whose \mbox{inter-agent} communication structure is described by a communication graph. The global STL specification is formulated as the conjunction of local task specifications assigned to individual agents and selected agent pairs. This formulation naturally captures applications such as formation control, containment, and coverage, where coordinated interactions are essential to accomplish shared objectives. In this setting, \mbox{inter-agent} collaboration requirements are encoded through a task graph which, as defined earlier, is not required to coincide with the underlying communication graph.
To overcome the mismatch between the communication and task graphs, we propose a decentralized \mbox{$k$-hop} Prescribed Performance State Observer (\mbox{$k$-hop} PPSO) that enables each agent to estimate the states of task-relevant but \mbox{non-communicating} agents. The proposed observer guarantees convergence of the state estimation errors to a prescribed neighborhood of the origin while enforcing predefined performance bounds during the transient. In contrast to our earlier works in \cite{11018605} and \cite{zaccherini2025robustestimationcontrolheterogeneous}, which proposed a \mbox{$k$-hop} \mbox{observer-based} control framework for standard \mbox{multi-agent} objectives, the present work integrates the observer design with Prescribed Performance Control to ensure satisfaction of Signal Temporal Logic (STL) specifications. Compared to classical control, where the input can be designed to be bounded or with bounded derivative, PPC under STL specifications does not provide any boundedness guarantee before proof of task satisfaction. Thus, to guarantee prescribed performances on the state estimations, additional assumptions are required to decouple the observer and the feedback controller. Exploiting this decoupling, we show that the spatial robustness of a broad class of STL specifications can be systematically adjusted to account for \mbox{worst-case} estimation errors. This enables control of the MAS toward task satisfaction even during the transient phase of the observer dynamics.

To the best of our knowledge, this work represents the first decentralized \mbox{estimation-based} controller for MAS subject to STL specifications under communication–task graph mismatch.

The contribution of the paper is summarized as follows:
\begin{itemize}
    \item Compared to \cite{ACIKMESE20141037, 6308694,8264361}, where distributed observers and \mbox{consensus-based} filters without performance guarantees are designed for homogeneous and undisturbed MAS, this work handles heterogeneous MAS subject to bounded external disturbances by extending our previous work in \cite{11018605}. Moreover, in contrast to the observer in \cite{zaccherini2025robustestimationcontrolheterogeneous}, which was developed for standard control objectives, the observer proposed here is specifically tailored for STL specifications.
    
    \item Unlike \cite{LINDEMANN2021100973} and \cite{10918825}, the method proposed here allows for \mbox{communication-task} graph mismatch. Specifically, using the \mbox{worst-case} estimation guarantees provided by the \mbox{$k$-hop} PPSO, we show that, under typical STL tasks, the spatial robustness can be adjusted to account for the \mbox{worst-case} estimation error. This leads to an equivalent condition that, once imposed, guarantees task satisfaction even when state estimates are used instead of the true states.

     \item A preliminary investigation of the proposed methodology appeared in \cite{zaccherini2026observerbasedcontrolmultiagentsystems}. 
     Compared to this earlier work, the present paper introduces revised constraints on the prescribed performance functions, leading to less conservative initialization requirements. Moreover, while in \cite{zaccherini2026observerbasedcontrolmultiagentsystems} the observer design was only briefly outlined due to space limitations and presented as a variation of the approach in \cite{zaccherini2025robustestimationcontrolheterogeneous}, here we provide a complete analysis and formally prove that the state estimation errors satisfy predefined prescribed performance bounds within the PPC framework.
     Finally, to further demonstrate scalability, the simulation study is extended to a \mbox{large-scale} MAS composed of omnidirectional robots.
\end{itemize}

%% file: Content/Preliminaries.tex
\textbf{Notation:} 
Let $\mathbb{R}$, $\mathbb{R}_{\ge 0}$, and $\mathbb{R}_{>0}$ denote the sets of real, nonnegative real, and strictly positive real numbers, respectively. Furthermore, let $\mathbb{N}$ denote the set of natural numbers including $0$. The Euclidean space of dimension $n$ is $\mathbb{R}^n$, and $\mathbb{R}^{n\times m}$ represents the space of real matrices with $n$ rows and $m$ columns. The identity and the null matrices of dimension $n$ are $I_n$ and $0_n$, and $1_n$ denotes the $n$-dimensional vector of ones.
For any set $S$, $|S|$ denotes its cardinality, $S^c$ its complement, and $\partial S$ its boundary. For $N$ sets $\{S_1,\dots, S_N\}$, the Cartesian product, intersection, and union are $\bigtimes^N _{i=1} S_i$, $\bigcap^N _{i=1} S_i$, and $\bigcup^N _{i=1} S_i$, respectively. For a set $S = \{s_1,\ldots,s_n\}$ with $s_i \in \mathbb{R}$, $S_{\max} := \max_{i\in\{1,\ldots,n\}} s_i$ and $S_{\min} := \min_{i\in\{1,\ldots,n\}} s_i$. 
$\top$ and $\bot$ denote the logical values true and false.
Let $B \in \mathbb{R}^{n\times n}$ be a symmetric matrix. Its smallest and largest eigenvalues are $\lambda_{\min}(B)$ and $\lambda_{\max}(B)$. We write $B \succ 0$ if $B$ is positive definite, and $\|B\|$ to denote its spectral norm. For $B \in \mathbb{R}^{n\times m}$, $(B)_{i,j}$ denotes the \mbox{$(i,j)$-th} entry, and $(B)_{i,:} \in \mathbb{R}^{1\times m}$ the \mbox{$i$-th} row. For $x \in \mathbb{R}^n$, the Euclidean norm is $\|x\| = \sqrt{x^\top x}$.
The diagonal matrix with diagonal entries $a_1,\ldots,a_n$ is $\mathrm{diag}(a_1,\ldots,a_n)$, and $\otimes$ denotes the Kronecker product. A function $f$ is of class $\mathcal{C}^1$ if it is continuously differentiable on its domain.

\subsection{Signal Temporal Logic (STL)}
\label{Section Signal Temporal Logic}
Signal Temporal Logic (STL) is a formal framework for specifying temporal and quantitative properties of \mbox{continuous-time} signals \cite{maler2004monitoring}. It integrates logical predicates with temporal operators to express constraints on both signal values and the time intervals over which these constraints must be satisfied. Each predicate $\mu$ is defined through a function $\mathcal{P} : \mathbb{R}^n \rightarrow \mathbb{R}$, $\mathcal{P} \in \mathcal{C}^1$, such that $\mu = \top$ if $\mathcal{P}(x) \ge 0$ and $\mu = \bot$ otherwise, for $x \in \mathbb{R}^n$.
In this paper we consider \mbox{non-temporal} \eqref{eq: definition of non temporal formula} and temporal formulas \eqref{eq: definition of temporal formula}, defined recursively as:
\begin{equation}
    \label{eq: definition of non temporal formula}
    \psi ::= \top \ | \  \mu \  | \ \neg \mu \ | \ \psi_1 \land \psi_2,
\end{equation}
\begin{equation}
    \label{eq: definition of temporal formula}
    \phi ::= G_{[a,b]} \psi \ | \ F_{[a,b]} \psi  \ | \  F_{[\underline{a},\underline{b}]} G_{[\bar{a},\bar{b}]}  \psi, 
\end{equation}
where $\neg$, $\land$, $ G_{[a,b]}$, and $F_{[a,b]}$ denote the negation, conjunction, always and eventually operators, with $a, b \in \mathbb{R}_{\geq 0}$ and $ a \leq b$. $\psi_1$ and $\psi_2$ in \eqref{eq: definition of non temporal formula} are \mbox{non-temporal} formulas of class \eqref{eq: definition of non temporal formula}.

To quantify the degree to which a trajectory $x: \mathbb{R}_{\geq 0} \rightarrow X \subseteq \mathbb{R}^n$ satisfies or violates an STL formula $\phi$, we adopt the notion of robust semantics from \cite[Def.~3]{donze2010robust}. Formally, the robust semantics of STL operators are defined as 
\begin{equation*}
    \begin{split}
        \rho^{\mu}(x,t)&:= \mathcal{P}(x(t)), \\
        \rho^{\neg \phi}(x,t)&:= -\rho^{\phi}(x,t),    \\
         \rho^{\phi_1 \wedge \phi_2}(x,t)&:= \min (\rho^{\phi_1}(x,t), \rho^{\phi_2}(x,t)), \\
        \rho^{F_{[a, b]} \phi}(x,t)&:=\max _{t_1 \in [t+a, t+b]} \rho^\phi(x, t_1),\\
        \rho^{G_{[a, b]} \phi}(x,t)&:= \min _{t_1 \in [t+a, t+b]} \rho^\phi(x, t_1),\\
        \rho^{F_{[\underline{a},\underline{b}]} G_{[\bar{a},\bar{b}]}  \phi}(x,t) &:=  \max _{t_1 \in [t+\underline{a}, t+\underline{b}]} \min _{t_2 \in [t_1+\bar{a}, t_1+\bar{b}]} \rho^\phi(x, t_2),
    \end{split}
\end{equation*}
where $\rho^\phi(x, t): X \times \mathbb{R}_{\geq 0}  \rightarrow \mathbb{R}$ is the \textit{robustness function} of $\phi$, such that $\rho^\phi(x, t) > 0$ implies that $\phi$ is satisfied at time $t$. To develop a \mbox{continuous-time} feedback controller that requires the derivative of the robustness function (cf. Theorem \ref{Theorem on task satisfacion}), $\rho^{\phi_1 \wedge \phi_2}(x,t)$ is approximated by the smooth function $\bar{\rho}^{\phi_1 \wedge \phi_2}(x,t) = -\frac{1}{\eta} \ln(\exp(- \eta \rho^{\phi_1}(x,t)) + \exp(- \eta \rho^{\phi_2}(x,t)))$ as in \cite{7799279}. Note that, since $\bar{\rho}^{\phi_1 \wedge \phi_2}(x,t) \leq \rho^{\phi_1 \wedge \phi_2}(x,t)$ for all $\eta >0$, with equality when $\eta \rightarrow \infty$, $\bar{\rho}^{\phi_1 \wedge \phi_2}(x,t) > 0$ implies $\rho^{\phi_1 \wedge \phi_2}(x,t)> 0$.

\subsection{Multi-agent systems}
\label{Sec: MAS, communication graoh and task graph}
Consider a team of $N \in \mathbb{N}$ heterogeneous agents $\mathcal{V} = \{1,\dots, N\}$. Denote with $\vect{x}$ and $\vect{u}$ the global state and input of the system and suppose each agent $i \in \mathcal{V}$ behaves as:
\begin{equation}
    \label{eq: agent's dynamic}
    \dot x_i(t) =  f_i(x_i(t)) + g_i(x_i(t))u_i(t) + w_i(\vect{x},t),
\end{equation}
where $x_i \in \mathbb{R}^{n_i}$ and $u_i \in \mathbb{R}^{m_i}$ are the state and input of agent $i$, respectively, $f_i:\mathbb{R}^{n_i} \rightarrow \mathbb{R}^{n_i}$ is the flow drift, $g_i:\mathbb{R}^{n_i} \rightarrow \mathbb{R}^{n_i \times m_i}$ is the input matrix, and $w_i: \mathbb{R}^{\sum^N _{i=1} n_i} \times \mathbb{R}_{\geq 0} \rightarrow \mathbb{R}^{n_i}$.

\begin{remark}
    $w_i$ is a function of the global state $\vect{x}$ and time $t$. Thus, it can represent external disturbances acting on $i$ or coupling interactions among agents.
\end{remark}

Let $n = \sum_{i=1}^{N} n_i$ and $m = \sum_{i=1}^{N} m_i$ be the dimensions of the global state and input. Then, $\vect{x}= \left[x_1^\top, \dots, x_N^\top \right]^\top \in \mathbb{R}^n$ and $ \vect{u}=\left[u_1^\top, \dots, u_N^\top\right]^\top \in \mathbb{R}^m$.

\begin{assumption}
    \label{Assumption on existence of a solution}
    (i) $f_i:\mathbb{R}^{n_i} \rightarrow \mathbb{R}^{n_i}$ is unknown and locally Lipschitz; (ii) $g_i:\mathbb{R}^{n_i} \rightarrow \mathbb{R}^{n_i \times m_i}$ is locally Lipschitz and $g_i(x_i)g_i(x_i)^\top $ is positive definite for all $x_i \in \mathbb{R}^{n_i}$; (iii) $w_i: \mathbb{R}^{n} \times \mathbb{R}_{\geq 0} \rightarrow \mathbb{R}^{n_i}$ is unknown, continuous and uniformly bounded on $\mathbb{R}^{n} \times \mathbb{R}_{\geq 0}$.
\end{assumption}

\textbf{Communication graph:} The \mbox{information-exchange} among the agents is described by an undirected graph $\mathcal{G}_C := (\mathcal{V}, \mathcal{E}_C)$, where $\mathcal{E}_C \subseteq \mathcal{V} \times \mathcal{V}$ represents the set of communication links among the agents $\mathcal{V}$. A path between two agents $i,j \in \mathcal{V}$ is defined as a sequence of non-repeating edges connecting $i$ to $j$. A \mbox{$k$-hop} path between $i$ and $j$ corresponds to a path of length $k$, i.e., a sequence of exactly $k$ edges connecting $i$ to $j$. Given $i,j\in \mathcal{V}$, $D_{ij}$ denotes the distance between $i$ and $j$, which is defined by the length of the shortest path in $\mathcal{G}_C$ connecting $i$ and $j$. 
For each agent $i \in \mathcal{V}$, we define with $\Nhop{i}{k} := \{ j \in \mathcal{V} \,|\, \exists \, p\text{-hop path from $j$ to $i$ with } 2 \leq p \leq k \}$, $\Neigh{i}{C} := \{ j \in \mathcal{V} \,|\, (i,j) \in \mathcal{E}_C \}$ and $\Neighext{i}{C} := \Neigh{i}{C}\cup \{i\}$ the set of $k$-hop neighbors, direct neighbors and extended neighbors of $i$, respectively. We label the elements of $\Nhop{i}{k}$ as $\{N^i_1, \dots, N^i_{\eta_i}\}$, where each $N^i_j \in \mathcal{V}$ denotes the global index of the \mbox{$j$-th} \mbox{$k$-hop} neighbor of $i$, and $\eta_i = |\Nhop{i}{k}|$ represents the total number of such neighbors. 

\begin{assumption}
    \label{Assumption on communication graph and neighbors}
    $\mathcal{G}_C$ is a \mbox{time-invariant}, connected and undirected graph, and each $i\in \mathcal{V}$ knows $\Neigh{i}{C}$ and $\Nhop{i}{k}$.
\end{assumption}
\begin{assumption}
    \label{Assumption on 2-hop communication}
    Each agent $i\in \mathcal{V}$ has access and can propagate $x_j(t)$, for all $j \in \Neigh{i}{C}$, to $\Neigh{i}{C}$ at all $t \in \mathbb{R}_{>0}$.
\end{assumption}

Assumption~\ref{Assumption on communication graph and neighbors} is not restrictive, since distributed neighborhood discovery algorithms have been extensively studied in the sensor network literature \cite{994183}. Furthermore, Assumption~\ref{Assumption on 2-hop communication} is satisfied in scenarios where each agent can measure the states of its direct neighbors using onboard sensors and can share this information with its direct neighbors.
\begin{remark}
    \label{Remark on validity of Assumption 3}
    Assumption~\ref{Assumption on 2-hop communication} may appear restrictive. However, as will be discussed in Remark~\ref{Remark on local disagreement vector computation and positive definiteness of Mik}, the proposed method remains applicable even when this assumption is violated by incorporating the elements of $\Neigh{i}{C}$ into the \mbox{$k$-hop} neighborhood, i.e., by introducing $\Bar{\mathcal{N}}_i^{k\text{-hop}} := \Nhop{i}{k} \cup \Neigh{i}{C}$. Note that, by the definition of $\Neigh{i}{C}$, agent~$i$ is also excluded from the redefined neighborhood $\Bar{\mathcal{N}}_i^{k\text{-hop}}$.
\end{remark}

\textbf{Task dependency graph: } The MAS in \eqref{eq: agent's dynamic} is subject to a global STL specification $\phi =\land_{i=1}^N \phi_i$, where each $\phi_i$ denotes the local STL specification assigned to agent $i \in \mathcal{V}$.
Task dependencies among agents are represented by the \textit{task dependency graph} $\mathcal{G}_T := (\mathcal{V}, \mathcal{E}_T)$, where an edge $(i,j) \in \mathcal{E}_T$ if the robustness function $\rho^{\phi_i}$ depends explicitly on the state $x_j$, i.e., satisfaction of $\phi_i$ depends on agent $j$. The task neighborhood of agent $i$ is defined as $\mathcal{N}^T_i := \{ j \in \mathcal{V} \,|\, (i,j) \in \mathcal{E}_T \lor j = i\}$, i.e., it contains the agents whose state affects the satisfaction of task $\phi_i$.

\begin{assumption}
    \label{Assumption on Task graph}
   $\mathcal{G}_T$ is a directed acyclic graph, where \mbox{self-loops} are excluded from the definition of cycles and are therefore permitted. Consequently, each task $\phi_i$ may depend on its local state $x_i$.
\end{assumption}

Note that, since the tasks are assigned to the agents at the design stage, Assumption~\ref{Assumption on Task graph} is not limiting.

\subsection{State estimates}
\label{Section: state and input estimate}
Let $\vect{x}^i$ denote the stacked vector containing the state of the \mbox{$k$-hop} neighbors of agent $i$, i.e., of $N_j^i \in \Nhop{i}{k}$:
\begin{equation}
    \label{stack vector of real values estimated by agent i}
    \state{\vect{x}}{i}{} := \left[\state{x}{\top}{N_1^i}, \dots ,\state{x}{\top}{N^i_{\eta_i}} \right]^\top
\end{equation}
and let $\estimate{\vect{x}}{i}{} := \left[\estimate{x}{i \ \top}{N_1^i}, \dots ,\estimate{x}{i \ \top}{N^i_{\eta_i}} \right]^\top$ be its estimate carried out by the agent~$i$. Moreover, denote with  $\tilde{\vect{x}}^i$ the corresponding estimation error, i.e.:
\begin{equation}
    \label{Definition: error on input and state estimation perfromed by i}
    \error{\vect{x}}{i}{} :=\left[\error{x}{i \ \top}{N_1^i}, \dots ,\error{x}{i \ \top}{N^i_{\eta_i}}\right]^\top,
\end{equation}
where $\error{x}{i}{N_j^i} := \estimate{x}{i}{N_j^i} -\state{x}{}{N_j^i}$ for all $N_j^i \in \Nhop{i}{k}$.

Define $\vect{x}_i := 1_{\eta_i} \otimes x_i$ and let
\begin{equation}
    \label{eq: estimation of state and input of agent i}
    \estimate{\vect{x}}{}{i} := \left[\estimate{x}{N_1^i \top}{i}, \dots , \estimate{x}{N_{\eta_i}^i \top}{i} \right]^\top
\end{equation}
be the stacked vector containing the estimates of $x_i$ computed by the \mbox{$k$-hop} neighbors of agent $i$, i.e., $\estimate{x}{N_j^i}{i}$ is the estimate of $x_i$ performed by agent $N_j^i \in \Nhop{i}{k}$, with $j \in \{1, \dots, \eta_i\}$.

As in \eqref{Definition: error on input and state estimation perfromed by i}, denote with $\error{\vect{x}}{}{i} := \estimate{\vect{x}}{}{i} -\vect{x}_i$ the estimation error computed by each $N_j^i \in \Nhop{i}{k}$, i.e.:
\begin{equation}
    \label{Definition: error on input and state estimation of agent i}
        \hspace{-0.2cm}\error{\vect{x}}{}{i} := \left[\error{x}{N_1^i \top}{i}, \dots ,\error{x}{N_{\eta_i}^i \top}{i} \right]^\top,
\end{equation}
with $\error{x}{N_j^i}{i} := \estimate{x}{N_j^i}{i} - x_i$ for all $N_j^i \in \Nhop{i}{k}$.

For notational simplicity, assume without loss of generality that $n_i = 1$ for all $i \in \mathcal{V}$. The extension to \mbox{higher-dimensional} cases follows directly via the Kronecker product.

\subsection{Individual/collaborative tasks}
\label{Section: collaborative and non-collaborative tasks}
Consider the local task $\phi_i$ assigned to agent $i \in \mathcal{V}$. As discussed in Section~\ref{Section Signal Temporal Logic}, the degree of satisfaction of $\phi_i$ is quantified by the robustness function $\rho^{\phi_i}(\vect{x}_{\phi_i}, t)$, where $\vect{x}_{\phi_i}$ collects the state of all agents participating in $\phi_i$, i.e., those $x_j$ for which $j\in \Neigh{i}{T}$.

Satisfaction of $\phi_i$ may depend on agent $i$ and possibly other agents in $\mathcal{V}\setminus \{i\}$. Accordingly, we denote $\phi_i$ as \textbf{individual} if $\Neigh{i}{T} = \{i\}$, and \textbf{collaborative} otherwise.
A collaborative specification $\phi_i$ may involve agents that are directly connected to agent $i$ through the communication graph $\mathcal{G}_C$, as well as agents that are not, namely agents $j \in \Neigh{i}{T} \cap \Neighext{i}{C}$ and $j \in \Neigh{i}{T} \setminus \Neighext{i}{C}$, respectively. Consequently, due to mismatches between the task and communication graphs, the vector $\vect{x}_{\phi_i}$ may be only partially available to agent $i$. To address this issue, we introduce a local estimate $\hat{\vect{x}}_{\phi_i}$ and define the corresponding robustness estimate as $\rho^{\phi_i}(\hat{\vect{x}}_{\phi_i}, t)$.
When only partial state information is available, robustness evaluated at the estimated state defines a different function than the robustness evaluated at the true state, i.e., $\rho^{\phi_i}(\hat{\vect{x}}_{\phi_i},t) \not\equiv \rho^{\phi_i}(\vect{x}_{\phi_i},t)$. In contrast, under full state knowledge, the two functions coincide, i.e., $\rho^{\phi_i}(\hat{\vect{x}}_{\phi_i},t) \equiv \rho^{\phi_i}(\vect{x}_{\phi_i},t) $.

For notational purposes, let $\Bar{\vect{x}}_{\phi_i}$ denote the components of $\hat{\vect{x}}_{\phi_i}$ corresponding to known states, i.e., $x_j$ with $j \in \Neighext{i}{C} \cap \Neigh{i}{T}$, and let $\estimate{\vect{x}}{i}{\phi_i}$ denote the components corresponding to estimated states, i.e, $x_j$ with $j \in \Neigh{i}{T} \setminus \Neighext{i}{C}$. Define $N^{T}_{i} := \sum_{j \in \Neigh{i}{T}} n_j$, $N^{TC}_{i} := \sum_{j \in \Neighext{i}{C}\cap \Neigh{i}{T}} n_j$, and $N^{T\setminus C}_{i} := \sum_{j \in \Neigh{i}{T} \setminus \Neighext{i}{C}} n_j$. Then, $\vect{x}_{\phi_i}, \hat{\vect{x}}_{\phi_i}\in \mathbb{R}^{N^{T}_{i}}$, $\Bar{\vect{x}}_{\phi_i} \in \mathbb{R}^{N^{TC}_{i}}$ and $\estimate{\vect{x}}{i}{\phi_i} \in \mathbb{R}^{N^{T\setminus C}_{i}}$. 

\subsection{Prescribed performance functions}
The notion of prescribed performance function \cite{4639441} is formally defined below.
\begin{definition}
    \label{Definition of Prescribed performance function}
    $\rho :\mathbb{R}_{\geq 0} \rightarrow \mathbb{R}_{> 0}$ is a prescribed performance function if it satisfies, for all $t \in \mathbb{R}_{\geq 0}$: (i) $\rho \in \mathcal{C}^1$; (ii)  $\rho(t) \leq \overline{\rho}$ for some $ \overline{\rho} <  \infty $ and (iii) $|\dot{\rho}(t)|\leq \overline{\rho}_d$ for some $\overline{\rho}_d<\infty$.
\end{definition}

One conventional choice of prescribed performance function is the decreasing exponential function 
\begin{equation}
    \label{eq: choice of prescribed performance function}
    \rho(t)= (\rho(0)-\rho(\infty)) e^{-l t} + \rho(\infty),
\end{equation}
where $\rho(0)$ and $\rho(\infty)$ denote the initial and steady-state values, and $l > 0$ specifies the decay rate.
\begin{lemma}
    \label{Lemma: sum of prescribed perofmance functions}
    Let $\rho_i (t)$, $i \in \{1, \dots n\}$, be prescribed performance functions as per Definition~\ref{Definition of Prescribed performance function}. Then $\sum^n_{i = 1}\rho_i (t)$ is also a prescribed performance function.
\end{lemma}
\begin{proof}
    (i) Since $\rho_i \in \mathcal{C}^1$ for all $i \in \{1,\dots, n\}$, $\sum^n_{i = 1}\rho_i (t) \in \mathcal{C}^1$; (ii) From Definition~\ref{Definition of Prescribed performance function}, $\rho_i (t) \leq \overline{\rho}_i$ holds for some bounded $\overline{\rho}_{i} \in \mathbb{R}_{> 0}$. Hence, $\sum^n_{i = 1}\rho_i (t) \leq \sum^n_{i = 1} \overline{\rho}_{i}$; (iii) $|\dot{\rho}_i(t)|\leq \overline{\rho}_{d,i}$ holds for some $\overline{\rho}_{d,i} < \infty$. Thus $|\sum^n_{i = 1}\dot{\rho}_i(t)| \leq \sum^n_{i = 1} \overline{\rho}_{d,i}$.
\end{proof}
\begin{lemma}
    \label{Lemma om the norm of prescibed performance functions}
    Let $\vect{\rho} (t) = \bigl[\rho_1(t), \dots , \rho_m(t)\bigr]^\top$ be a vector whose components $\rho_i$, $i \in \{1, \dots, m\}$, are prescribed performance functions as per Definition \ref{Definition of Prescribed performance function}. Then, $\norm{\vect{\rho}(t)}$ is also a prescribed performance function satisfying conditions (i)-(iii) in Definition \ref{Definition of Prescribed performance function}.
\end{lemma}
\begin{proof}
    (i) Each $\rho_i(t)$ is positive and continuously differentiable ($\mathcal{C}^1$) by definition, hence $\vect{\rho} \in \mathcal{C}^1$. The Euclidean norm is smooth in $\mathbb{R}^{\eta_i} \setminus \{0\}$. Furthermore, $\vect{\rho}(t) \neq 0$ for all $t \ge 0$. Thus, $\norm{\vect{\rho}} \in \mathcal{C}^1$. (ii) From Definition~\ref{Definition of Prescribed performance function}, $0 < \rho_i(t) \leq \overline{\rho}_i$ holds with $\overline{\rho}_i < \infty$. Hence, $\norm{\vect{\rho}(t)} \leq \overline{\vect{\rho}}$, with $\overline{\vect{\rho}} = \sqrt{ \sum_{i=1}^{m}(\overline{\rho}_i)^2}$. (iii) Differentiating the Euclidean norm gives
     $\frac{d\norm{\vect{\rho}(t)}}{dt} = \frac{\vect{\rho}^\top(t) \dot{\vect{\rho}}(t)}{\norm{\vect{\rho}(t)}}$. Then, since $\norm{\vect{\rho}(t)} > 0$, $0 <\rho_i(t) <\infty$, and $|\dot{\rho}_i(t)| <\infty$ for all $i \in \{1, \dots, m\}$, it follows that $\frac{d\norm{\vect{\rho}(t)}}{dt}$ is upper bounded.
\end{proof}

\subsection{Problem statement}
\label{Section: Problem statement}
We now have all the ingredients to provide a formal definition of the problem considered in this paper:
\begin{problem}
    \label{Problem: first problem formulation}
    Consider a heterogeneous MAS \eqref{eq: agent's dynamic} with a communication graph $\mathcal{G}_C$ and a global STL specification $\phi = \land_{i=1}^N \phi_i$, where each $\phi_i$, defined as in \eqref{eq: definition of temporal formula}, denotes a local STL task assigned to agent $i$. Let $\mathcal{G}_T$ be the task dependency graph induced by $\phi$, and suppose Assumptions~\ref{Assumption on existence of a solution}-\ref{Assumption on Task graph} hold. Then, for each $i \in \mathcal{V}$, design a decentralized observer that estimates the state $x_j$ of agents $j \in \Neigh{i}{T} \setminus \Neighext{i}{C}$, while ensuring that the estimation error satisfies prescribed performance bounds. Furthermore, design a decentralized \mbox{observer-based} control input $u_i$ that guarantees satisfaction of the global specification $\phi$. To guarantee a decentralized framework, both the observer and the controller must rely only on locally available information, i.e., the agent’s own information and information exchanged with its direct neighbors.
\end{problem}

The remainder of the paper is organized as follows. Section~\ref{Sec: Decentralized State Observer}
introduces a decentralized $k$-hop PPSO that enables each agent $i \in \mathcal{V}$ to estimate ${x}_{N_j^i}$, for all $N_j^i \in \Nhop{i}{k}$, while ensuring that the state estimation errors satisfy prescribed performance bounds. Section~\ref{Sec: Observer-based Task Satisfaction} introduces the notion of \mbox{cluster-induced} graph and, building on the observer’s performance guarantees, reformulates the STL tasks to account for the \mbox{worst-case} estimation errors. Section~\ref{Sec: Distributed Controller Design} proposes an \mbox{observer-based} decentralized controller that ensures task satisfaction also when state estimations are used instead of the true state. Section~\ref{Sec: Case Study} presents simulation results, and Section~\ref{Sec: Conclusions} concludes the paper with remarks and directions for future work.

%% file: Content/Observer.tex
This section introduces the concepts of \textit{\mbox{$k$-hop} induced graphs} and \textit{disagreement vectors}, and presents a decentralized observer that allows each agent $i\in \mathcal{V}$ to estimate the states of its \mbox{$k$-hop} neighborhood $\Nhop{i}{k}$. Specifically, building on the PPC framework~\cite{4639441}, we propose a decentralized \mbox{$k$-hop} Prescribed Performance State Observer ($k$-hop PPSO) that ensures the convergence of the state estimation error to a neighborhood of the origin while guaranteeing prescribed transient performance specified a priori.

\subsection{$\mathit{k}$-hop induced graphs}
For every $i \in \mathcal{V}$, let $\mathcal{G}^{\khop{k}}_i$ denote the \textit{\mbox{$k$-hop} induced graph} associated to $i$, i.e., the graph $\mathcal{G}^{\khop{k}}_i := ( \Nhop{i}{k}, \mathcal{E}^{\khop{k}}_i)$ induced by the \mbox{$k$-hop} neighbors $\Nhop{i}{k}$ of agent $i$, where $\mathcal{E}^{\khop{k}}_i := \{(p,q) \in \mathcal{E}_C: \{p,q\} \subseteq \Nhop{i}{k}\}$. 
 
For every agent $i\in \mathcal{V}$, let $\kc{M}{i} \in \mathbb{R}^{\eta_i \times \eta_i}$ be the matrix defined as
\begin{equation}
    \label{eq: Mkc defintion}
    \kc{M}{i} := \kc{L}{i} + \kc{H}{i},
\end{equation}
where $\kc{L}{i}$ is the Laplacian matrix \cite[Eq. 2.9]{mesbahi2010graph} associated to $\mathcal{G}^{\khop{k}}_i$, and $\kc{H}{i}:= \text{diag}(|\Neigh{N_1^i}{C}\cap \Neigh{i}{C}|, \dots, |\Neigh{N_{\eta_i}^i}{C}\cap \Neigh{i}{C}|) \in \mathbb{N}^{\eta_i \times \eta_i}$.
Since $\mathcal{G}^{\khop{k}}_i$ may be disconnected and split into multiple connected \mbox{sub-components}, the positive definiteness of $\kc{M}{i}$ is studied in the following.

From the definition of \mbox{$k$-hop} neighbors in Section \ref{Sec: MAS, communication graoh and task graph}, the following lemmas hold for the \mbox{$k$-hop} induced graphs.
\begin{lemma}
    \label{lemma: neighbors}
    Under Assumption \ref{Assumption on communication graph and neighbors}:
    (i) $|\Neigh{j}{C} \cap \Nhop{i}{k}|> 0 \  \text{ or } \ |\Neigh{j}{C}\cap \Neigh{i}{C}|> 0$ for all $i \in \mathcal{V}$ and $j \in \Nhop{i}{k}$.
    (ii) For all $i \in \mathcal{V}$ and for each connected component in the \mbox{subgraph} $\mathcal{G}^{\khop{k}}_i$, there exists at least one $j \in \Nhop{i}{k}$ for which $\ |\Neigh{j}{C}\cap \Neigh{i}{C}|> 0$.
\end{lemma}
\begin{proof}
    (i) By contradiction assume there exists $i \in \mathcal{V}$ and $j \in \Nhop{i}{k}$ such that $|\Neigh{j}{C}\cap \Nhop{i}{k}|= 0$ and $|\Neigh{j}{C}\cap \Neigh{i}{C}| = 0$. If $|\Neigh{j}{C}\cap \Neigh{i}{C}| = 0$, then $D_{ij} > 2$, which implies $|\Neigh{j}{C}\cap \Nhop{i}{k}|\neq 0$, contradicting the original assumption. Similarly, if $|\Neigh{j}{C}\cap \Nhop{i}{k}|= 0$ with $j \in \Nhop{i}{k}$, then $D_{ij} = 2$. Therefore, there must exist $q \in \Neigh{j}{C}\cap \Neigh{i}{C}$, which contradicts the assumption that $|\Neigh{j}{C}\cap \Neigh{i}{C}| = 0$.
    (ii) Suppose by contradiction that for some $i \in \mathcal{V}$, every $j$ in a connected component of $\mathcal{G}^{\khop{k}}_i$ is such that $|\Neigh{j}{C}\cap \Neigh{i}{C}|= 0$. This implies there does not exist a \mbox{$l$-hop} path, with $2\leq l\leq k$, between $i$ and any agent $q$ in the connected component, which contradicts the assumption $q \in \Nhop{i}{k}$.
\end{proof}

Intuitively, condition (i) in Lemma~\ref{lemma: neighbors} ensures that every agent in the \mbox{$k$-hop} neighborhood of $i \in \mathcal{V}$ is connected either directly to $i$ via a shared neighbor or to another agent within $\Nhop{i}{k}$. Condition (ii) further guarantees that each connected component of $\mathcal{G}^{\khop{k}}_i$ contains at least one agent sharing a direct neighbor with $i$. Together, these conditions ensure that every connected component remains indirectly linked to $i$, allowing information originating from $i$ to propagate across the entire \mbox{$k$-hop} induced graph.

\begin{lemma}
    \label{Lemma: sum of two positive semi-definite matrices}
    (\cite[Cor. 4.3.12]{Horn_Johnson_2012})
    Let $A, B \in \mathbb{R}^n$ be Hermitian and suppose that $B$ is positive \mbox{semi-definite}. Then, $\lambda_i(A) \leq \lambda_i(A+B)$, $i =1,\dots, n$, with equality for some $i$ if and only if $B$ is singular and there is a nonzero vector $x$ such that $A x = \lambda_i(A)x$, $Bx = 0$, and $(A+B)x = \lambda_i(A+B)x$.
\end{lemma}

Using Lemmas~\ref{lemma: neighbors} and \ref{Lemma: sum of two positive semi-definite matrices}, it follows that:
\begin{lemma}
    \label{Lemma: Mkc positive definit}
     Under Assumption \ref{Assumption on communication graph and neighbors}, $\kc{M}{i} \succ 0 $ for all $i \in \mathcal{V}$.
\end{lemma}

\begin{proof}
Since $\kc{L}{i}$ is the Laplacian matrix of the subgraph $\mathcal{G}^{\khop{k}}_i$, we have $\lambda_{\min}(\kc{L}{i}) = 0$. Furthermore, since $\kc{M}{i}$ is the sum of two real, positive \mbox{semi-definite}, symmetric matrices, it is Hermitian \cite[Def.~4.1.1]{Horn_Johnson_2012}. Therefore, Lemma~\ref{Lemma: sum of two positive semi-definite matrices} applies, and $\kc{M}{i} \succ 0$ is guaranteed if there does not exist an $x \in \mathbb{R}^{\eta_i}$ orthogonal to $\kc{L}{i}$ and $\kc{H}{i}$, i.e., for which $\kc{L}{i} x = 0$ and $\kc{H}{i} x = 0$. 
Suppose $\mathcal{G}^{\khop{k}}_i$ contains $m^{\text{sub}}_i$ connected components $\mathcal{G}^{\khop{k}}_{i, r}$, $r \in \{1,\dots, m^{\text{sub}}_i\}$. Then, each eigenvector associated with the zero eigenvalue of $\kc{L}{i}$ belongs to the span of the $m^{\text{sub}}_i$ binary vectors representing consensus on each connected component $\mathcal{G}^{\khop{k}}_{i, r}$. However, by Lemma \ref{lemma: neighbors}, each $\mathcal{G}^{\khop{k}}_{i, r}$ has at least an associated \mbox{non-zero} element in the diagonal matrix $\kc{H}{i}$. Thus, none of the eigenvectors of $\kc{L}{i}$ is orthogonal to $\kc{H}{i}$, leading to $0 < \lambda_j(\kc{M}{i})$ for all $j \in \{1, \dots, \eta_i\}$.
\end{proof}

The following example is provided to clarify the notation and definitions above.
\begin{example}
    Consider the MAS in Fig~\ref{fig: example of khop induced graph, communication graph}. The \mbox{$3$-hop} induced graph associated to Agent $1$ is shown in Fig~\ref{fig: example of khop induced graph, 2hop induced graph} and is defined as $\mathcal{G}_1^{\khop{k}} = (\Nhop{1}{k}, \mathcal{E}_1^\khop{k})$, where $\Nhop{1}{k} = \{3,4,5,6\}$ and $\mathcal{E}_1^\khop{k} = \{(3,5), (4,6)\}$. 

As a result, $\kc{L}{i} = \begin{bmatrix}
    I_2 & -I_2 \\ -I_2 & I_2
    \end{bmatrix}$, $\kc{H}{i} = \begin{bmatrix}
     I_2 & 0_2 \\ 0_2 & 0_2
    \end{bmatrix}$ and $\kc{M}{i} = \begin{bmatrix}
    2I_2 & -I_2 \\ -I_2 & I_2
    \end{bmatrix}.$
    $\hfill \triangle$ 
\end{example}
\begin{figure}[t]
    \centering
    \begin{subfigure}[b]{0.48\linewidth}
        \centering
        \includegraphics[width=\linewidth]{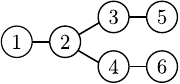}
        \caption{}
        \label{fig: example of khop induced graph, communication graph}
    \end{subfigure}
    \hfill
    \begin{subfigure}[b]{0.48\linewidth}
        \centering
        \includegraphics[width=0.45\linewidth]{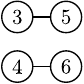}
        \caption{}
        \label{fig: example of khop induced graph, 2hop induced graph}
    \end{subfigure}
    \caption{(a) Communication graph; (b) Agent~$1$'s $3$-hop induced graph.}
    \label{fig: example of khop induced graph}
\end{figure}

\subsection{Disagreement vector}
For each $i \in \mathcal{V}$ and $N_j^i \in \Nhop{i}{k}$, define the disagreement component on the estimate of $x_i$ performed by $N_j^i$ as:
\begin{equation}
    \label{eq: disagreement vector on the estimate of agent i}
    \xi_i^{N_j^i} :=\sum_{l\in(\Neigh{N_j^i}{C} \cap \Nhop{i}{k})} (\statediff{\hat{x}}{N_j^i}{i}{\hat{x}}{l}) + |\Neigh{N_j^i}{C} \cap \Neigh{i}{C}|  (\hat{x}^{N_j^i}_i - x_i),
\end{equation}
where $\xi_i^{N_j^i}$ represents a local disagreement term capturing the inconsistency of the estimate $\hat{x}^{N_j^i}_i$ with respect to: (i) the state estimates $\estimate{x}{l}{i}$ shared by agents $l \in \Neigh{N^i_{j}}{C} \cap \Nhop{i}{k}$ and (ii) the true state information $x_i$ shared by agents $l \in \Neigh{N^i_{j}}{C} \cap \Neigh{i}{C}$. In particular, the first summation term in \eqref{eq: disagreement vector on the estimate of agent i} quantifies the disagreement with neighboring estimates, while the second term penalizes the deviation from $x_i$, weighted by the number of neighbors with access to this information, i.e., $|\Neigh{N_j^i}{C} \cap \Neigh{i}{C}|$. Hence, $\xi_i^{N_j^i}$ acts as a consensus-like correction term promoting alignment with neighboring estimates and consistency with locally available state information.

By stacking $\xi_i^{N_j^i}$ for all $N_j^i \in \Nhop{i}{k}$ and by using the state estimation error defined in \eqref{Definition: error on input and state estimation of agent i}, the \textbf{disagreement vector} defined as $\vect{\xi}_i := \Bigl[\xi_{i}^{N_1^i}, \dots ,\xi_{i}^{N^i_{\eta_i}}\Bigr]^\top$ can be expressed as:

\begin{equation}
    \label{eq: disagreement vector expression}
    \vect{\xi}_i = (\kc{L}{i} + \kc{H}{i}) \error{\vect{x}}{}{i} = \kc{M}{i} \error{\vect{x}}{}{i},
\end{equation}
where $\kc{M}{i}\in \mathbb{R}^{\eta_i \times \eta_i}$ is defined as in \eqref{eq: Mkc defintion}.

\begin{remark}
    \label{Remark on local disagreement vector computation and positive definiteness of Mik}
    If Assumption~\ref{Assumption on 2-hop communication} does not hold, $x_i$  becomes unavailable for local computation of the disagreement component. Nevertheless, \eqref{eq: disagreement vector on the estimate of agent i} can still be evaluated locally and the positive definiteness of $\kc{M}{i}$ can be preserved by extending the definition of the \mbox{$k$-hop} neighborhood to include each agent’s \mbox{$1$-hop} neighbors, i.e., by introducing $\Bar{\mathcal{N}}_i^{k\text{-hop}} := \Nhop{i}{k} \cup \Neigh{i}{C}$ for all $i \in \mathcal{V}$. In this case, $\kc{L}{i}$ is the Laplacian matrix of the subgraph induced by the new \mbox{$k$-hop} neighbors $\Bar{\mathcal{N}}_i^{k\text{-hop}}$ of agent~$i$, and $\kc{H}{i} := \text{diag}(h^{N_1^i}_i, \dots, h^{N_{\eta_i}^i}_i) \in \mathbb{R}^{\eta_i \times \eta_i}$, where $h^{N_j^i}_i = 1$ if $N_j^i \in \Neigh{i}{C}$.
\end{remark}

\subsection{Prescribed performance observer}
\label{Section: Observer requirements}
To guarantee prescribed performances on the state estimation errors, the observer must be designed so that $\error{x}{N_j^i}{i}(t)$ satisfies
\begin{equation}
    \label{eq: state estimation error inequality}
    |\error{x}{N_j^i}{i}(t)| <\delta^{N_j^i}_i(t),
\end{equation}
for all $i \in \mathcal{V}$ and $N_j^i \in \Nhop{i}{k}$, where $\delta^{N_j^i}_i :\mathbb{R}_{\geq 0} \rightarrow \mathbb{R}$ is a prescribed performance function, defined as in Definition~\ref{Definition of Prescribed performance function}, which specifies the desired time-varying bounds on the corresponding estimation error.

Lemma~\ref{Lemma: Mkc positive definit} ensures that $\kc{M}{i}$ is invertible. Hence, from \eqref{eq: disagreement vector expression} we obtain $\error{\vect{x}}{}{i} = {\kc{M}{i}}^{-1} \vect{\xi}_i$. Consequently, since $|{\error{x}{N_j^i}{i}}| \leq \sum^{\eta_i}_{r = 1} |{{(\kc{M}{i}}^{-1}})_{N_j^i, r}| |\xi^{r}_i|$ holds, to enforce \eqref{eq: state estimation error inequality} it suffices to constrain the evolution of $\xi_i^{N_j^i}$ such that
\begin{equation}
    \label{eq: definition of the bounds on the disagreement term}
    |\xi_i^{N_j^i}| < \rho_i^{N_j^i}(t)
\end{equation}
holds for all $N_j^i \in \Nhop{i}{k}$, where each $\rho_i^{N_{j}^i}(t)$ is a prescribed performance function satisfying
\begin{equation}
    \label{eq: constraint on the prescribed prformance functions}
    \sum^{\eta_i}_{r = 1} |{{(\kc{M}{i}}^{-1}})_{N_j^i, N_r^i}| \rho^{N_r^i}_i(t) \leq \delta^{N_{j}^i}_i(t),  
\end{equation}
$\rho^{N_j^i}_i(0) \geq \rho^{N_j^i}_i(\infty) > 0$ and $\rho^{N_j^i}_i(0) > |\xi_{i}^{N_j^i}(0)|$.
Lemma~\ref{Lemma: sum of prescribed perofmance functions} implies that $\sum^{\eta_i}_{r = 1} |{{(\kc{M}{i}}^{-1}})_{N_j^i, N_r^i}| \rho^{N_r^i}_i(t)$ is a prescribed performance function. Therefore, $\rho_i^{N_{j}^i}(t)$ can be designed so that \eqref{eq: constraint on the prescribed prformance functions} is satisfied with equality for some $N_j^i \in \Nhop{i}{k}$.

In our preliminary work \cite{zaccherini2026observerbasedcontrolmultiagentsystems}, the performance functions $\rho_i^{N_{j}^i}(t)$ were designed to satisfy \begin{equation}
    \label{eq: old constraint on the prescribed prformance functions}
    \lVert\vect{\rho}_i(t)\rVert  \leq \lambda_{\text{min}}(\kc{M}{i}) \min_{j \in \{1,\dots, \eta_i\}} \{\delta_i^{N_j^i}(t)\},
\end{equation}
with $\vect{\rho}_i(t) := \Bigl[\rho_i^{N_1^i}(t),\dots , \rho_i^{N_{\eta_i}^i}(t)\Bigr]^\top$. As will be shown in the following, this design is quite conservative. In particular, for fixed bounds $\delta_i^{N_{j}^i}(t)$, it permits smaller admissible values of $\rho_i^{N_{j}^i}(t)$ than the \mbox{component-wise} design proposed in this paper. Formally:

\begin{proposition}
    Any vector $\vect{\rho}_i(t)$ satisfying the norm constraint in \eqref{eq: old constraint on the prescribed prformance functions} also satisfies \eqref{eq: constraint on the prescribed prformance functions}. Furthermore, for fixed bounds $\delta_i^{N_{j}^i}(t)$, the \mbox{component-wise} constraint \eqref{eq: constraint on the prescribed prformance functions} admits larger $\rho_i^{N_{j}^i}(t)$ than those resulting from \eqref{eq: old constraint on the prescribed prformance functions}.
\end{proposition}
\begin{proof}
    Let $\mathcal{S}_1(t) := \{\vect{\rho}_i(t): \lVert\vect{\rho}_i(t)\rVert  \leq \lambda_{\text{min}}(\kc{M}{i}) \min_{j \in \{1,\dots, \eta_i\}} \{\delta_i^{N_j^i}(t)\}\}$ and $\mathcal{S}_2(t) := \{\vect{\rho}_i(t): \sum^{\eta_i}_{r = 1} |{{(\kc{M}{i}}^{-1}})_{N_j^i, r}| \rho^{r}_i(t) \leq \delta^{N_{j}^i}_i(t), \forall N_{j}^i \in \Nhop{i}{k}\}$ be the sets containing the vectors $\vect{\rho}_i(t)$ satisfying the constraints in \eqref{eq: old constraint on the prescribed prformance functions} and \eqref{eq: constraint on the prescribed prformance functions}, respectively.
    
    Since prescribed performance functions are positive by Definition~\ref{Definition of Prescribed performance function},  the \mbox{Cauchy–Schwarz} inequality implies $\sum^{\eta_i}_{r = 1} |{{(\kc{M}{i}}^{-1}})_{N_j^i, r}| \rho^{r}_i \leq \sqrt{\sum^{\eta_i}_{r = 1} |{{(\kc{M}{i}}^{-1}})_{N_j^i, r}|^2}\norm{\vect{\rho}_i(t)}$ for all $i \in \mathcal{V}$ and $N^i_j \in \Nhop{i}{k}$. By definition of Euclidean norm, $\sqrt{\sum^{\eta_i}_{r = 1} |{{(\kc{M}{i}}^{-1}})_{N_j^i, r}|^2} = \norm{( {\kc{M}{i}}^{-1})_{N_j^i,:}}$. Thus, since $\norm{( {\kc{M}{i}}^{-1})_{N_j^i,:}} \leq \norm{{\kc{M}{i}}^{-1}}$ and $\norm{{\kc{M}{i}}^{-1}} \leq \lambda_{\text{min}}(\kc{M}{i})^{-1} $, the inequality $\sum^{\eta_i}_{r = 1} |{{(\kc{M}{i}}^{-1}})_{N_j^i, r}| \rho^{r}_i \leq  \lambda_{\text{min}}(\kc{M}{i})^{-1}  \norm{\vect{\rho}_i(t)}$ is satisfied. For any $\vect{\rho}_i(t) \in \mathcal{S}_1$, $\lVert\vect{\rho}_i(t)\rVert  \leq \lambda_{\text{min}}(\kc{M}{i}) \min_{j \in \{1,\dots, \eta_i\}} \{\delta_i^{N_j^i}(t)\}$ holds. Thus, $\sum^{\eta_i}_{r = 1} |{{(\kc{M}{i}}^{-1}})_{N_j^i, r}| \rho^{r}_i \leq \min_{j \in \{1,\dots, \eta_i\}} \{\delta_i^{N_j^i}(t)\} \leq \delta_i^{N_{j}^i}(t)$ is valid for all $\vect{\rho}_i(t) \in \mathcal{S}_1(t)$.
    Then, since $\lVert\vect{\rho}_i(t)\rVert  \leq \lambda_{\text{min}}(\kc{M}{i}) \min_{j \in \{1,\dots, \eta_i\}} \{\delta_i^{N_j^i}(t)\}$ implies $\sum^{\eta_i}_{r = 1} |{{(\kc{M}{i}}^{-1}})_{N_j^i, r}| \rho^{r}_i(t) \leq \delta^{N_{j}^i}_i(t)$, we conclude $\mathcal{S}_1(t) \subseteq \mathcal{S}_2(t)$.
\end{proof}

\begin{remark}
    Since $\rho_i^{N_j^i}(t)$ are design choices, desired behavior can be imposed to every estimation error $\error{x}{N_j^i}{i}(t)$ by designing appropriate $\rho_i^{N_j^i}(t)$.
\end{remark}

Inspired by the PPC literature \cite{4639441}, for every $i \in \mathcal{V}$ and $ N_j^i \in \Nhop{i}{k}$, we introduce the normalized and the transformed disagreement as $e_i^{N_j^i} :=\rho^{N_j^i}_i(t)^{-1}\xi_{i}^{N_j^i}$ and  
\begin{equation}
\vspace{-0.05cm}
    \label{eq: transformed normalized disagreement}
    \epsilon^{N_j^i}_i := T(e^{N_j^i}_i) = T(\rho^{N_j^i}_i(t)^{-1}\xi_{i}^{N_j^i}),
\end{equation}
where $T: (-1,1) \rightarrow \mathbb{R}$ is a strictly increasing transformation satisfying $T(0) = 0$. In this work, we select $T(e) = \ln(\frac{1+e}{1-e})$, which has a strictly positive derivative $ J_T(e) = \frac{2}{1-e^2}$.
Then, by defining the transformed normalized disagreement vector as $\vect{\epsilon}_{i}  := \left[\epsilon^{N_1^i}_i, \dots ,\epsilon_{i}^{N^i_{\eta_i}} \right]^\top$, its dynamics result into:
\begin{equation}
    \label{eq: disagreement vector dynamics on agent i state}
     \dot{\vect{\epsilon}}_i = \vect{J}_i  \vect{P}^{-1}_{i}(\dot{\vect{\xi}}_{i} - \dot{\vect{P}}_{i} \vect{e}_i),
\end{equation}
where $\vect{J}_i := \text{diag}\left(J_T(e^{N_1^i}_i), \dots, J_T(e^{N^i_{\eta_i}}_i)\right)$, $\vect{P}_{i} := \text{diag}\left(\rho^{N_1^i}_i, \dots, \rho^{N_{\eta_i}^i}_i\right)$, $\dot{\vect{P}}_{i} := \text{diag}\left(\dot{\rho}^{N_1^i}_i, \dots, \dot{\rho}^{N^i_{\eta_i}}_i\right)$, $\vect{e}_i := \vect{P}^{-1}_i \vect{\xi}_i$ and $\dot{\vect{\xi}}_{i} := \Bigl[\dot{\xi}_{i}^{N_1^i}, \dots ,\dot{\xi}_{i}^{N^i_{\eta_i}}\Bigr]^\top$.

\begin{remark}
    \label{Remark: boundness E impliess funnel satisfaction}
    From \eqref{eq: transformed normalized disagreement}, it follows that if the vector $\vect{\epsilon}_{i}$ is bounded, then $e^{N_j^i}_i$ remains confined within the interval $(-1,1)$ for all $N_j^i \in \Nhop{i}{k}$. Consequently, for every $i\in \mathcal{V}$ and each $N^i_j \in \Nhop{i}{k}$, $\xi_{i}^{N^i_j}$ evolves in compliance with \eqref{eq: definition of the bounds on the disagreement term}.
\end{remark}
For the \mbox{multi-dimensional} case, where $n_i \neq 1$, this reasoning can be performed on every component of the agent's state. Hence, desired performance can be imposed on the convergence of every disagreement component of $\xi_{i}^{N_j^i}$, i.e., on every $\xi_{i,l}^{N_j^i}$, with $l \in \{1, \dots, n_i\}$.

\subsection{Observer design}
\label{Sec: Observer design}
To design observers that satisfy the prescribed performance requirements in \eqref{eq: definition of the bounds on the disagreement term}, we adopt a variation of the decentralized \mbox{$k$-hop} PPSO proposed in \cite{zaccherini2025robustestimationcontrolheterogeneous}. Specifically, for all $i \in \mathcal{V}$ and $N_j^i \in \Nhop{i}{k}$, we design an observer whose state $\estimate{x}{i}{N_j^i}$ evolves according to
\begin{equation}
    \label{eq: component of the state estimation error dyanmics}
    \dot{\hat{x}}_{i}^{N_j^i} = - \rho_i^{N_j^i}(t)^{-1} J_T(e_i^{N_j^i}) \epsilon_{i}^{N_j^i}(t),
\end{equation}
where $J_T(e^{N_j^i}_i)$, $e^{N_j^i}_i$ and $\epsilon^{N_j^i}_{i}(t)$ are defined as in \eqref{eq: transformed normalized disagreement}.
\begin{remark}
    Note that $\epsilon_{i}^{N_j^i}(t)$, and consequently $\dot{\hat{x}}_{i}^{N_j^i}(t)$, is computed exclusively based on information received from the neighbors of agent $N_j^i$. Hence, the proposed observer is decentralized.
\end{remark}
By stacking $\dot{\hat{x}}_{i}^{N_j^i}$ for all $N_j^i \in \Nhop{i}{k}$, the dynamics of $\estimate{\vect{x}}{}{i}$, defined as in \eqref{eq: estimation of state and input of agent i}, becomes:
\begin{equation}
    \label{Eq: State observer dynamics}
    \begin{split}
        \dot{\hat{\vect{x}}}_i &= -  \vect{P}_i^{-1} \vect{J}_i \vect{\epsilon}_i,
    \end{split}
\end{equation}
where $\vect{P}_i$, $\vect{J}_i$, and $\vect{\epsilon}_i$ are defined as {in \eqref{eq: disagreement vector dynamics on agent i state}}.

At the purpose of the observer convergence, consider the following assumption:
\begin{assumption}
    \label{Assumption: observer asssumptions for convergence}
    For all $i\in \mathcal{V}$ and $t \in \mathbb{R}_{\geq 0}$, the input $u_i(t) $ is bounded and designed to enforce $x_i(t) \in \mathcal{X}_i$, where $\mathcal{X}_i \subset \mathbb{R}^{n_i}$ is a bounded set.
\end{assumption}
As will be shown in Theorem~\ref{Theorem on task satisfacion}, a decentralized controller $u_i$ can be designed to satisfy Assumption~\ref{Assumption: observer asssumptions for convergence} in the considered framework, thus Assumption~\ref{Assumption: observer asssumptions for convergence} is not restrictive in the given context.

With the ingredients introduced so far, the observer convergence can be stated as follows.
\begin{theorem}
    \label{Theorem: main theorem on state estimation convergence}
    Consider a heterogeneous MAS \eqref{eq: agent's dynamic} with connected communication graph $\mathcal{G}_C$ and decentralized state observers as in \eqref{eq: component of the state estimation error dyanmics}. Under Assumptions~\ref{Assumption on existence of a solution}-\ref{Assumption on 2-hop communication} and \ref{Assumption: observer asssumptions for convergence}, the state estimation error $\tilde{x}^{N^i_j}_{i}(t)$ satisfies $|\tilde{x}_{i}^{N_j^i}(t)| < \delta^{N_j^i}_i(t)$ for all $N_j^i \in \Nhop{i}{k}$ and all $i\in \mathcal{V}$, provided that $|\xi^{N^i_j}_{i}(0)| < \rho^{N_j^i}_i(0)$ holds, and $\rho^{N_j^i}_i(t)$ is designed so that $\sum^{\eta_i}_{r = 1} |{{(\kc{M}{i}}^{-1}})_{N_j^i, r}| \rho^{r}_i \leq \delta^{N_{j}^i}_i(t)$.
\end{theorem}

\begin{proof}
    Consider agent $i \in \mathcal{V}$. According to Assumption~\ref{Assumption on communication graph and neighbors}, $\mathcal{G}_C$ is a \mbox{time-invariant} graph. Thus, $\kc{M}{i}$ is constant, $\dot{\vect{\xi}}_i = \kc{M}{i} \dot{\tilde{\vect{x}}}_{i}$ from \eqref{eq: disagreement vector expression},  and \eqref{eq: disagreement vector dynamics on agent i state} can be rewritten as:
    \begin{equation}
        \label{eq: transformed error dynamics}
        \dot{\vect{\epsilon}}_i = \vect{J}_i \vect{P}^{-1}_{i}(\kc{M}{i} \dot{\tilde{\vect{x}}}_{i} - \dot{\vect{P}}_{i} \vect{e}_i).
    \end{equation}
    From the agent's dynamics in \eqref{eq: agent's dynamic}, the definitions in \eqref{Definition: error on input and state estimation of agent i}, and the observer \eqref{Eq: State observer dynamics}, $\dot{\tilde{\vect{x}}}_{i}$ becomes $\dot{\tilde{\vect{x}}}_{i} = -\vect{f}_i(x_{i}) - \vect{g}_i(x_{i}) u_i - \vect{w}_i(\vect{x},t) -  \vect{P}_i^{-1} \vect{J}_i \vect{\epsilon}_i $, where $\vect{f}_i(x_i) := 1_{\eta_i} \otimes f_i(x_i)$, $\vect{g}_i(x_i) := 1_{\eta_i} \otimes g_i(x_i)$ and $\vect{w}_i(\vect{x},t) := 1_{\eta_i} \otimes w_i(\vect{x},t)$.
    
     Consider the candidate Lyapunov function $V(\vect{\epsilon}_i(\vect{e}_i)) = \frac{1}{2} \vect{\epsilon}_i(\vect{e}_i)^T \vect{\epsilon}_i(\vect{e}_i)$, with time derivative $\dot{V} = \vect{\epsilon}_i^T \dot{\vect{\epsilon}}_i$. By replacing \eqref{eq: transformed error dynamics} and $\dot{\tilde{\vect{x}}}_{i}$, $\dot{V} $ results into:
     \begin{equation}
        \label{eq Lyapunvon function derivative}    
        \hspace{-0.3cm}
        \begin{split}
            \dot{V} =& - \vect{\epsilon}_i^T \vect{J}_i \vect{P}^{-1}_{i} \kc{M}{i} \vect{P}_i^{-1} \vect{J}_i \vect{\epsilon}_i - \vect{\epsilon}_i^T \vect{J}_i \vect{P}^{-1}_{i} \Bigl\{\dot{\vect{P}}_{i} \vect{e}_i + \\& \kc{M}{i} [ \vect{f}_i(x_{i}) + \vect{g}_i(x_{i}) u_i + \vect{w}_i(\vect{x},t)] \Bigr\}.
        \end{split}
    \end{equation} 
     Since $\kc{M}{i} \succ 0$ from Lemma~\ref{Lemma: Mkc positive definit}, $- \vect{\epsilon}_i^T \vect{J}_i \vect{P}^{-1}_{i} \kc{M}{i} \vect{P}_i^{-1} \vect{J}_i \vect{\epsilon}_i \leq - \lambda_{\min}(\kc{M}{i}) \alpha_{J}\alpha_{\rho} \vect{\epsilon}_i^T \vect{\epsilon}_i$ holds with $\alpha_{J} = \min_{N_j^i \in \Nhop{i}{k}} \bigl\{\min_{e^{N_j^i}_i \in (-1,1)} J_T(e^{N_j^i}_i)^2\bigr\}= 4$ and $\alpha_{\rho} = \max_{N_j^i \in \Nhop{i}{k}}\{\max_{t \in \mathbb{R}_{\geq 0}} \rho_i^{N_j^i}(t)^2\}$. From Definition~\ref{Definition of Prescribed performance function}, there exists $\overline{\rho}^{N_j^i}_i < \infty$ such that $\rho_i^{N_j^i}(t) \leq \overline{\rho}^{N_j^i}_i$. Thus, $\alpha_{\rho}$ is bounded as $0 < \alpha_{\rho} \leq \max_{N_j^i \in \Nhop{i}{k}} \{(\overline{\rho}^{N_j^i}_i)^2\}$.
     By summing and subtracting $\zeta \norm{\vect{P}^{-1}_{i} \vect{J}_i \vect{\epsilon}_{i}}^2$ for some $0 < \zeta < \lambda_{\min}(\kc{M}{i})$, \eqref{eq Lyapunvon function derivative} can be upper bounded as
     \begin{equation}
        \label{eq: Lyapunov fucntion final upper bound observer}
         \begin{split}
             \dot{V} \leq &-(\lambda_{\min}(\kc{M}{i}) - \zeta)\alpha_{J} \alpha_{\rho} \norm{\vect{\epsilon}_i}^2 - \vect{\epsilon}_i^T \vect{J}_i \vect{P}^{-1}_{i}b(t) \\&-\zeta \norm{\vect{P}^{-1}_{i} \vect{J}_i \vect{\epsilon}_{i}}^2,
         \end{split}
     \end{equation}
     where $b(t) = \dot{\vect{P}}_{i} \vect{e}_i + \kc{M}{i} [ \vect{f}_i(x_{i}) + \vect{g}_i(x_{i}) u_i + \vect{w}_i(\vect{x},t)]$.
     Since $\vect{\epsilon}_i^T \vect{J}_i \vect{P}^{-1}_{i}b(t) + \zeta \norm{\vect{P}^{-1}_{i} \vect{J}_i \vect{\epsilon}_{i}}^2$ are terms of the quadratic form $\norm{\sqrt{\zeta} \vect{P}^{-1}_{i} \vect{J}_i \vect{\epsilon}_{i} + \frac{1}{2\sqrt{\zeta}}b(t)}^2$, and $\norm{\sqrt{\zeta} \vect{P}^{-1}_{i} \vect{J}_i \vect{\epsilon}_{i} + \frac{1}{2\sqrt{\zeta}}b(t)}^2 \geq 0$ holds by the definition of norm, the inequality $- \vect{\epsilon}_i^T \vect{J}_i \vect{P}^{-1}_{i}b(t) -\zeta \norm{\vect{P}^{-1}_{i} \vect{J}_i \vect{\epsilon}_{i}}^2 \leq \frac{1}{4\zeta}b^\top(t)b(t)$ is satisfied for all $\vect{\epsilon}_i$, and $\dot{V}$ can be bounded as $\dot{V} \leq -(\lambda_{\min}(\kc{M}{i}) - \zeta)\alpha_{J} \alpha_{\rho} \norm{\vect{\epsilon}_i}^2 + \frac{1}{4\zeta}b^\top(t)b(t)$. Thus, we can write
     \begin{equation}
        \label{eq: Lyapunov function final upper bound, observer}
        \dot{V} \leq -\kappa V + \vect{b}(t),
    \end{equation}
    with $\kappa := 2(\lambda_{\min}(\kc{M}{i}) - \zeta)\alpha_{J} \alpha_{\rho}$ and $\vect{b}(t) := \frac{1}{4\zeta} \{\lambda_{\max}(\kc{M}{i})[\norm{\vect{f}_i(x_i)} + \norm{\vect{g}_i(x_i)} \norm{u_i} + \norm{\vect{w}_i}] + \norm{\dot{\vect{P}}_{i} \vect{e}_i}\}^2$. To proceed, let's check whether $\vect{b}(t)$ admits an upper bound $\Bar{\vect{b}}(t)$.   

    From Assumption~\ref{Assumption on existence of a solution}-(i), $f_i$ and $g_i$ are Lipschitz continuous functions. Thus, given (iii) of Assumption~\ref{Assumption on existence of a solution} and Assumption~\ref{Assumption: observer asssumptions for convergence}, $\norm{\vect{f}_i(x_i)}$, $\norm{\vect{g}_i(x_i)} \norm{u_i}$ and $\norm{\vect{w}_i}$ are bounded.
    Let $\tilde{\mathcal{X}}_i(t) = \{ \error{\vect{x}}{}{i} \in \mathbb{R}^{\eta_i}| -1_{\eta_i} <\vect{e}_i = \vect{P}^{-1}_i \vect{\xi}_i <1_{\eta_i}\}$ be the time varying set containing the state estimation errors $\error{\vect{x}}{}{i}$ for which the disagreement terms $\xi_{i}^{N_j^i}(t)$ satisfy the bounds \eqref{eq: definition of the bounds on the disagreement term} for all $N_j^i \in \Nhop{i}{k}$. Then, $ \dot{\rho}^{N_j^i}_i e^{N_j^i}_i$ is bounded in $\tilde{\mathcal{X}}_i(t)$, and $|\dot{\rho}^{N_j^i}_i(t) e^{N_j^i}_i| < |\bar{\rho}^{N_j^i}_{d,i}(t)| < \bar{\rho}^{N_j^i}_{d,i}(t)$ holds  with $\bar{\rho}^{N_j^i}_{d,i}(t) < \infty$.
    As a result, $\norm{\dot{\vect{P}}_{i} \vect{e}_i}$ is bounded and an upper bound $\Bar{\vect{b}}(t) < \infty$ on $\vect{b}(t)$ is guaranteed to exist for all $ \error{\vect{x}}{}{i} \in \tilde{\mathcal{X}}_i(t)$.
    
    Inspired by \cite[Thm. 22]{10918825}, to prove the invariance of the set $\tilde{\mathcal{X}}_i(t)$, we introduce $S(\vect{e}_i) := 1 - \exp(- V(\vect{\epsilon}_i(\vect{e}_i)))$. Note that, from its definition, $S$ satisfies: (i) $S(\vect{e}_{i}) \in (0,1)$ for all $\vect{e}_{i} \in \mathcal{D}$ with $\mathcal{D} = \bigtimes_{j=1}^{\eta_i}(-1,1)$, and (ii) $S(\vect{e}_{i}) \rightarrow 1$ as $\vect{e}_{i} \rightarrow \partial \mathcal{D} $. Therefore, studying the boundedness of $\vect{\epsilon}_{i}(\vect{e}_i)$ through the one of $V$ reduces to proving that $S(\vect{e}_{i}) < 1$ holds for all $t$.  
    By replacing  \eqref{eq: Lyapunov fucntion final upper bound observer} and $V(\vect{\epsilon}_i(\vect{e}_i)) = - \ln (1-S(\vect{e}_{i}))$ in $\dot{S}(\vect{e}_{i}) = \dot{V}(\vect{\epsilon}_i(\vect{e}_i))(1-S(\vect{e}_{i}))$, $ \dot{S}(\vect{e}_{i}) \leq - \kappa (1-S(\vect{e}_{i})) \Bigl(-\frac{1}{\kappa}\vect{b}(t) - \ln (1-S(\vect{e}_{i}))\Bigr)$ is obtained.
    Since $\kappa >0$ and $1-S(\vect{e}_{i}) >0$ by definition, to verify whether $\dot{S}(\vect{e}_{i}) \leq 0 $ holds, it suffices to study under which conditions $-\frac{1}{\kappa}\vect{b}(t) - \ln (1-S(\vect{e}_{i}))\geq 0$ is valid.
    Note that $-\frac{1}{\kappa}\vect{b}(t) - \ln (1-S(\vect{e}_{i}))\geq 0 $ is satisfied for all $\vect{e}_{i} \in \Omega^c_{\vect{e}}$, where $\Omega_{\vect{e}}=\bigl\{\vect{e}_{i} \in \mathcal{D}| S(\vect{e}_{i}) < 1- \exp(-\frac{{\vect{b}}(t)}{\kappa}) \bigr\}$, and that $-\frac{1}{\kappa}\vect{b}(t) - \ln (1-S(\vect{e}_{i})) = 0 $ holds for $\vect{e}_{i} \in  \partial\Omega_{\vect{e}}$. Thus, $\dot{S}(\vect{e}_{i}) \leq 0$  for $\vect{e}_{i} \in \Omega^c_{\vect{e}}$, with $\dot{S}(\vect{e}_{i}) =0$ iff $\vect{e}_{i} \in \partial\Omega_{\vect{e}}$. 
    Since the initialization satisfies $|\xi^{N^i_j}_{i}(0)| < \rho^{N_j^i}_i(0)$ for all $N_j^i \in \Nhop{i}{k}$, it follows that $e^{N_j^i}_{i}(0) \in (-1,1)$ for all $N_j^i \in \Nhop{i}{k}$, and therefore that $S(\vect{e}_{i}(0)) < 1$. As a result, since $\exp(-\frac{{\vect{b}}(t)}{\kappa}) \geq \exp(-\frac{\Bar{\vect{b}}(t)}{\kappa}) > 0$ by definition, $S(\vect{e}_{i})) < 1 $ is preserved for all $t \in \mathbb{R}_{\geq 0}$, independently of whether $\vect{e}_{i}$ is initialized in $\Omega_{\vect{e}}$ or not.
    From the inequality $S(\vect{e}_{i}) < 1 $, boundedness of $V(\vect{\epsilon}_i(\vect{e}_i))$, and therefore of $\vect{\epsilon}_{i}$, follows. As a result, \eqref{eq: definition of the bounds on the disagreement term} is satisfied.
    
    If all $\rho^{N_j^i}_i(t)$ are designed so that \eqref{eq: constraint on the prescribed prformance functions} holds, $|\error{x}{N_j^i}{i}(t)| <  \delta^{N_j^i}_i(t)$ is guaranteed by construction for all $N_j^i \in \Nhop{i}{k}$ as explained in Section \ref{Section: Observer requirements}.
\end{proof}

\begin{corollary}
   Under the observer guarantees in Theorem~\ref{Theorem: main theorem on state estimation convergence}, a sufficient choice of $k$ ensuring that every agent $i \in \mathcal{V}$ can estimate the states of all agents $j \in \Neigh{i}{T}\setminus \Neighext{i}{C}$ is
\begin{equation}
    \label{selection of the parameter k}
    k \geq \max_{i \in \mathcal{V}, j \in (\mathcal{N}^T_i \setminus \Neighext{i}{C})} D_{ij},
\end{equation}
where $D_{ij}$ is the distance between agent $i$ and $j$ as defined in Section~\ref{Sec: MAS, communication graoh and task graph}.
\end{corollary}

Intuitively, $k$ must be chosen no smaller than the maximum shortest-path distance in $\mathcal{G}_{C}$ between every collaborating but non-communicating pair of agents $(i,j)$, with $i\in \mathcal{V}$ and $j \in \mathcal{N}^T_i \setminus \Neighext{i}{C}$.

%% file: Content/Observer_based_Task_Satisfaction.tex
In this section, the notion of \textit{\mbox{cluster-induced} graph} is first introduced. Then, following the approach in \cite{LINDEMANN2021100973}, the STL task is reformulated as a set of constraints on the spatial robustness defined in Section~\ref{Section Signal Temporal Logic}. Since the objective is the synthesis of an \mbox{observer-based} controller, a modified class of constraints is introduced to ensure STL task satisfaction under \mbox{worst-case} state estimation errors.

\subsection{Cluster-induced graph}
\label{Section: cluster induced graph}
To introduce the notion of cluster, let's define the intersection among graphs as follows.
\begin{definition}
    \label{def: graphs intersection}
    Given $M$ graphs $\mathcal{G}_i := (\mathcal{V}_i, \mathcal{E}_i)$, $i \in \{1,\dots, M\}$, the intersection graph is defined as $\mathcal{G}^{\cap} := (\mathcal{V}^{\cap}, \mathcal{E}^{\cap})$, where $\mathcal{V}^{\cap} := \bigcap_{i =1}^M \mathcal{V}_i$ and $ \mathcal{E}^{\cap} := \bigcap_{i =1}^M \mathcal{E}_i$. For the intersection among directed and undirected graphs, each edge $(j,l)$ of the undirected graph is considered as existing in both directions, i.e., both as $(j,l)$ and $(l,j)$.
\end{definition}

Then, the graph $\mathcal{G}_{CT} = (\mathcal{V} , \mathcal{E}_{CT}) := \mathcal{G}_C \cap \mathcal{G}_T$ is defined, according to Definition \ref{def: graphs intersection}, as the intersection of the communication and the task graphs, i.e., $\mathcal{E}_{CT} := \mathcal{E}_C \cap \mathcal{E}_T$.

The graph $\mathcal{G}_{CT}$ partitions the MAS into $N'$ connected components $\mathcal{C}_l := (\mathcal{V}_l, \mathcal{E}_l)$, with $l\in \mathcal{L}:=\{1,\dots, N'\}$, where $\mathcal{V}_l  := \{l_1, \dots, l_{v_l}\} \subseteq \mathcal{V}$ denotes the set of agents in $\mathcal{C}_l $, and $\mathcal{E}_l \subseteq \mathcal{E}_{CT}$ denotes the set of edges among them. We refer to $\mathcal{C}_l$, for all $l\in \mathcal{L}$, as \textbf{clusters}.
From the definition of \textit{connected components}, $\cup_{l = 1}^{N'} \mathcal{V}_l= \mathcal{V}$ and $\mathcal{V}_i \cap \mathcal{V}_j = \emptyset$ hold for all $i,j \in \mathcal{L}$ with $i \neq j$. As a result, each $\mathcal{C}_l$ is associated with $\phi^{c}_l = \land_{l_i \in \mathcal{V}_l} \phi_{l_i}$, and the global STL task $\phi$ can be rewritten as $\phi = \land_{l=1}^{N'} {\phi}^c_l$. 
For later use, let $\vect{x}_{\phi^c_l} := [x^{\top}_{l_1}, \dots x^{\top}_{l_{v_l}}]^{\top}$ denote the stacked state vector of the agents in cluster $\mathcal{C}_l$. Similarly, define with $\Hat{\vect{x}}_{\phi^c_l} := \Bigl[\estimate{\vect{x}}{{l_1} \ \top}{{\phi_{l_1}}}, \dots, \estimate{\vect{x}}{{l_{v_l}} \ \top}{{\phi_{l_{v_l}}}} \Bigr]^\top $ the vector collecting the state estimates, computed by each $l_i \in \mathcal{V}_l$, of the state of $j \in \Neigh{l_i}{T} \setminus \Neighext{l_i}{C}$, i.e., of the agents involved in $\phi_{l_i}$ but not in communication with $l_i$.

The \mbox{cluster-induced} graph is then defined as follows.
\begin{definition}
    \label{def: cluster induced graph}
    Given a MAS \eqref{eq: agent's dynamic}, with communication graph $\mathcal{G}_C$ and task graph $\mathcal{G}_T$, the \textit{\mbox{cluster-induced} graph} is defined as $\mathcal{G}' := (\mathcal{C}', \mathcal{E}')$, where each node in $\mathcal{C}'$ is associated to a cluster $\mathcal{C}_l$, with $l \in \mathcal{L}$, and an edge $(\mathcal{C}_l, \mathcal{C}_j) \in \mathcal{E}'$ exists if and only if there exists a task $\phi_{l_i}$, with ${l_i} \in \mathcal{V}_l$, whose satisfaction depends on the state of an agent $j_q \in \mathcal{V}_j$, with $(l_i,j_q) \in \mathcal{E}_T$.
\end{definition}
\begin{remark}
    \label{remark on the acyclicity of the cluster induced graph}
    Since $\mathcal{G}'$ is obtained by grouping the agents according to the clusters $\mathcal{C}_l$, no new directed edge is introduced. Thus, $\mathcal{G}'$ is a directed acyclic graph under Assumption~\ref{Assumption on Task graph}.
\end{remark}

For future reference, we define a \textit{leaf cluster} in the cluster-induced graph $\mathcal{G}'$ as a cluster with no outgoing edges in $\mathcal{E}'$, i.e., a cluster $\mathcal{C}_l$ such that $(\mathcal{C}_l, \mathcal{C}_z) \not\in \mathcal{E}'$ for all $z \in \mathcal{L} \setminus \{l\}$. Similarly, a cluster $\mathcal{C}_l$ is a \textit{root cluster} if it has no incoming edges, i.e., there exists no $\mathcal{C}_z$, with $z \in \mathcal{L} \setminus \{l\}$, such that $(\mathcal{C}_z, \mathcal{C}_l) \in \mathcal{E}'$.

To guarantee that all \mbox{intra-cluster} task dependencies match communication links, we introduce the following assumption.
\begin{assumption}
    \label{Assumption on the non-existance of a task not in communication inside a cluster}
     Task neighbors within the same cluster are assumed to be in communication. Formally, $\Neigh{l_i}{T} \cap \mathcal{V}_l \subseteq \Neighext{l_i}{C} \cap \mathcal{V}_l$ for all $l_i\in \mathcal{V}_l$ and $l \in \mathcal{L}$.
\end{assumption}

Since $\Neigh{l_i}{T} \cap \mathcal{V}_l \subseteq \Neigh{l_i}{T}$ and $\Neighext{l_i}{C} \cap \mathcal{V}_l \subseteq \Neighext{l_i}{C}$, Assumption~\ref{Assumption on the non-existance of a task not in communication inside a cluster} does not enforce $\Neigh{l_i}{T} \setminus \Neighext{l_i}{C} = \emptyset$, nor $\Neighext{l_i}{C} \setminus \Neigh{l_i}{T} = \emptyset$. Consequently, agent $l_i\in \mathcal{V}_l$ may still have tasks involving agents $j_q \in \mathcal{V}_j$, $j \in \mathcal{L}\setminus \{l\}$, with $j_q \in \Neigh{{l_i}}{T}$ and $j_q \notin \Neighext{{l_i}}{C}$. 

\begin{figure}[t!]
    \vspace{0.1cm}
    \centering
    \includegraphics[width=1\linewidth]{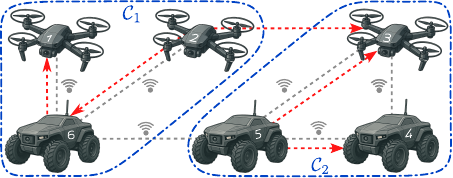}
    \caption{The graphs $\mathcal{G}_C$ and $\mathcal{G}_T$, respectively in gray and red, induce two clusters $\mathcal{C}_1$ and $\mathcal{C}_2$ with $\mathcal{V}_1 = \{1,2,6\}$ and $\mathcal{V}_2=\{3,4,5\}$.}
    \label{fig:Graph Gt, Gc, induced clusters}
    \vspace{-0.5cm}
\end{figure}

The following example is presented to clarify the \mbox{cluster-induced} graph definition.
\begin{example}
    Fig.~\ref{fig:Graph Gt, Gc, induced clusters} shows a MAS with $N = 6$ agents that communicate and collaborate according to the graphs $\mathcal{G}_C$ and $\mathcal{G}_T$  depicted in gray and red, respectively. $\mathcal{G}_C$ and $\mathcal{G}_T$ induces the MAS to be split in two clusters $\mathcal{C}_1 = ( \mathcal{V}_1, \mathcal{E}_1)$ and $\mathcal{C}_2 = ( \mathcal{V}_2, \mathcal{E}_2)$, where $\mathcal{V}_1 := \{1,2,6\}$ and $\mathcal{V}_2 := \{3,4,5\}$. Then, since $(2,3) \in \mathcal{E}_T$, the \mbox{cluster-induced} graph is defined as $\mathcal{G}' := (\mathcal{C}', \mathcal{E}')$, with $\mathcal{C}' = \{\mathcal{C}_1, \mathcal{C}_2\}$ and $\mathcal{E}'= \{(\mathcal{C}_1, \mathcal{C}_2)\}$.
    $\hfill \triangle$ 
\end{example}

\subsection{Observer-based task satisfaction}
\label{Section: observer-based task satisfaction}
Let $T_{\phi}^{\text{wdw}}$ denote the time window over which task $\phi$ must be satisfied.
Formally, for specifications involving the always operators $G_{[a,b]} \psi$, the corresponding time window is $T_{\phi}^{\text{wdw}} = [a,b]$. For eventually operators $F_{[a,b]} \psi$, the time window reduces to a single time instant, i.e., $T_{\phi}^{\text{wdw}} = \{t^*\}$, where $t^* \in [a,b]$ denotes the time at which the task is satisfied, in accordance with \cite{LINDEMANN2021100973}. For \mbox{eventually-always} operators $ F_{[\underline{a},\underline{b}]} G_{[\bar{a},\bar{b}]} \psi$, the time window is given by $T_{\phi}^{\text{wdw}} = [t^*+\bar{a},t^*+\bar{b}]$, where $t^* \in [\underline{a},\underline{b}]$ is selected analogously to the eventually operator, as in \cite{LINDEMANN2021100973}.

Then, as introduced in \cite{LINDEMANN2021100973}, to ensure satisfaction of the STL formula $\phi_i$, it suffices to constrain the \mbox{non-temporal} robustness $\rho^{\psi_i}(\vect{x}_{\phi_i})$ to satisfy:
\begin{equation}
    \label{eq: robustness inequality no estimations ivolved}
     -\gamma_{\psi_i}(t) + \rho_{\psi_i}^{\max} < \rho^{\psi_i}(\vect{x}_{\phi_i}) < \rho_{\psi_i}^{\max},
\end{equation}
where $\vect{x}_{\phi_i}$ collects the states involved in task $\phi_i$, and $\gamma_{\psi_i}(t)$ is a prescribed performance function designed such that $-\gamma_{\psi_i}(t) + \rho_{\psi_i}^{\max} >0$ holds for all $t \in T_{\phi_i}^{\text{wdw}}$.
A typical choice of the prescribed performance function is: $\gamma_{\psi_i}(t) = (\gamma_{\psi_i}^0 - \gamma_{\psi_i}^{\infty}) \exp(-l_{\psi_i} t) +\gamma_{\psi_i}^{\infty} $, where $\gamma^0_{\psi_i}$, $\gamma_{\psi_i}^{\infty}, l_{\psi_i} \in \mathbb{R}_{> 0}$ are design parameters selected as in \cite{LINDEMANN2021100973}. To guarantee feasibility of all task $\phi_i$, $i \in \mathcal{V}$, the following assumption is introduced:
\begin{assumption}
    \label{Assumption: assumption on task feasibility}
    The optimal robustness $\rho^{\text{opt}}_{\psi} := \text{sup}_{\vect{x}} \{\min_{i \in \mathcal{V}} \rho^{\psi_i}(\vect{x}_{\phi_i})\}$ satisfies $\rho^{\text{opt}}_{\psi} > 0$.
\end{assumption}
As discussed in Section~\ref{Section: collaborative and non-collaborative tasks}, the robustness $\rho^{\psi_i}$ may depend on estimated rather than actual states. As a result, the constraint in \eqref{eq: robustness inequality no estimations ivolved} cannot be enforced directly. However, the \mbox{$k$-hop} PPSO introduced in Section~\ref{Sec: Observer design} guarantees prescribed performance on the state estimation errors regardless of the system trajectory. This guarantee allows the formulation of an analogous constraint on $\rho^{\psi_i}(\hat{\vect{x}}_{\phi_i})$ that ensures task satisfaction under the \mbox{worst-case} estimation error. To this end, we introduce the following assumption.
\begin{assumption}
    \label{assumption on robustness lower bound}
    For all $i \in \mathcal{V}$, there exists a prescribed performance function $\rho_{\psi_i}^t:\mathbb{R}_{\geq 0} \rightarrow \mathbb{R}$, defined as in Definition \ref{Definition of Prescribed performance function}, such that
    \begin{equation}
        \label{eq: lower bound on robustness base on the estimation error}
        0 < {\rho}^{\psi_i}(\hat{\vect{x}}_{\phi_i}(t)) - \rho_{\psi_i}^t(t) \implies 0 < \rho^{\psi_i}(\vect{x}_{\phi_i}(t)).
    \end{equation} 
\end{assumption}

Assumption~\ref{assumption on robustness lower bound} is not restrictive in practice. Indeed, for each agent $i \in \mathcal{V}$: (i) $\vect{x}^i =  \estimate{\vect{x}}{i}{} - \error{\vect{x}}{i}{}$ holds by definition; (ii) Theorem~\ref{Theorem on task satisfacion} guarantees that $\norm{\error{\vect{x}}{i}{}} < \norm{[\delta^{i}_{N_1^i}(t), \dots, \delta^{i}_{N_{\eta_i}^i}(t)]}$ holds under the \mbox{$k$-hop} PPSO in Section \ref{Sec: Observer design}; and (iii) Lemma~\ref{Lemma om the norm of prescibed performance functions} ensures that $\norm{[\delta^{i}_{N_1^i}(t), \dots, \delta^{i}_{N_{\eta_i}^i}(t)]}$ is a prescribed performance function. Therefore, Assumption~\ref{assumption on robustness lower bound} holds for a large number of STL tasks.
For clarity, an example of a formation/containment task is presented below.
\begin{example}
    \label{example: funnel modification}
    Consider a path graph with $N=4$ agents. Denote the global state as $\vect{x}=[x_1 \ x_2\ x_3\ x_4]^\top \in \mathbb{R}^{Nn}$, where $x_i \in \mathbb{R}^{n}$, with $n= 1$, for all $i\in \{1,\dots, 4\}$. Assume $\Nhop{1}{3} =\{3,4\}$ and $\Neigh{1}{} =\{2\}$. Suppose agent $1$ is assigned a task $\phi_1$ with non-temporal robustness $\rho^{\psi_1}([x_1, x_3]^\top) := r_{13}^2 -\abs{x_1 - x_3 - d_{13}}^2$, where $r_{13}$ and $d_{13} \in \mathbb{R}_{\geq 0}$. Since $\state{x}{}{3} = \estimate{x}{1}{3} - \error{x}{1}{3}$ from definition,  
    \begin{equation}
        \label{eq in example: rewritten stl task}
         r_{13}^2 -\abs{x_1 - \estimate{x}{1}{3} + \error{x}{1}{3} - d_{13}}^2 > 0
    \end{equation}
    holds for all $t \in \mathbb{R}_{\geq 0}$. Then, since $\abs{x_1- \estimate{x}{1}{3} - d_{13} + \error{x}{1}{3}}^2 \leq (\abs{x_1 - \estimate{x}{1}{3}- d_{13}} + \abs{\error{x}{1}{3}})^2$ is valid from triangle inequality, \eqref{eq in example: rewritten stl task} is satisfied if $r_{13}^2 - (\abs{x_1 - \estimate{x}{1}{3}- d_{13}} + \abs{\error{x}{1}{3}})^2 > 0$. Since $r_{13}^2 - \abs{x_1 - \estimate{x}{1}{3}- d_{13}}^2 - \abs{\error{x}{1}{3}}^2 - 2\abs{x_1 - \estimate{x}{1}{3}- d_{13}}\abs{\error{x}{1}{3}} > 0$ holds in $\abs{x_1 - \estimate{x}{1}{3}- d_{13}} \in (0, - \abs{\error{x}{1}{3}} + r_{13})$, by introducing the bound resulting from the \mbox{$k$-hop} PPSO, i.e., $0 < \abs{\error{x}{1}{3}} < \delta_3^1(t)$,
    we derive that the previous inequality is satisfied if $\abs{x_1 - \estimate{x}{1}{3}- d_{13}} < - \delta_3^1(t) + r_{13}$, with $ \delta_3^1(t) < r_{13}$. By squaring both sides, we notice that $0 < \rho^{\psi_1}([x_1, \hat{x}^1_3]^\top) - \rho_{\psi_i}^t$ implies \eqref{eq in example: rewritten stl task} with $\rho^{\psi_1}([x_1, \hat{x}^1_3]^\top) = r_{13}^2 - \abs{x_1 - \estimate{x}{1}{3} -d_{13}}^2$ and $\rho_{\psi_1}^t(t) = 2 \delta_3^1(t)r_{13} -  \delta_3^1(t)^2$. Therefore, Assumption~\ref{assumption on robustness lower bound} holds.
    An analogous statement can be established for $n>1$ using the norm operator.
    $\hfill \triangle$ 
\end{example}

Under Assumption \ref{assumption on robustness lower bound}, ensuring satisfaction of $\phi_i$ in the presence of state estimation errors is equivalent to enforcing $-\gamma_{\psi_i}(t) + \rho^{\max}_{\psi_i} < \rho^{\psi_i}(\hat{\vect{x}}_{\phi_i}) - \rho_{\psi_i}^t(t) < \rho^{\max}_{\psi_i} -\rho_{\psi_i}^t(t)$, where $\gamma_{\psi_i}(t)$ is defined as in \eqref{eq: robustness inequality no estimations ivolved} and $\rho^{\psi_i}(\hat{\vect{x}}_{\phi_i})$ as in Section \ref{Section: collaborative and non-collaborative tasks}. By introducing $\Gamma_{\psi_i}(t) := \gamma_{\psi_i}(t) - \rho_{\psi_i}^t(t)$, the above requirement can be equivalently written as:
\begin{equation}
    \label{eq: inequality on robustness to guarantee task satisfaction, observer not in design stage of gamma}
    -\Gamma_{\psi_i}(t) < \rho^{\psi_i}(\hat{\vect{x}}_{\phi_i}) - \rho^{\max}_{\psi_i} < 0.
\end{equation}

For collaborative tasks $\phi_i$ for which $\Neigh{i}{T} \setminus \Neighext{i}{C} = \emptyset$, i.e., tasks involving only direct communicating agents, $\hat{\vect{x}}_{\phi_i} \equiv \vect{x}_{\phi_i}$ holds. Consequently, $\rho_{\psi_i}^t(t) = 0$ for all $t$ and condition \eqref{eq: inequality on robustness to guarantee task satisfaction, observer not in design stage of gamma} reduces to \eqref{eq: robustness inequality no estimations ivolved}.

\subsection{Prescribed performance functions design}
\label{Section: Performance functions parameter design}
As introduced in Section \ref{Section: observer-based task satisfaction}, we aim to enforce $0 < {\rho}^{\psi_i}(\hat{\vect{x}}_{\phi_i}(t)) - \rho_{\psi_i}^t(t)$ by imposing $-\gamma_{\psi_i}(t) + \rho^{\max}_{\psi_i} < \rho^{\psi_i}(\hat{\vect{x}}_{\phi_i}) - \rho_{\psi_i}^t(t) < \rho^{\max}_{\psi_i} -\rho_{\psi_i}^t(t)$ and through the proper design of $\gamma_{\psi_i}(t)$. Therefore, to ensure the satisfiability of $\phi_i$ regardless of $\hat{\vect{x}}_{\phi_i}$, the prescribed performance functions $\delta_{N^i_j}^i(t)$ and $\gamma_{\psi_i}(t)$  must be designed for all $i \in \mathcal{V}$ and $N^i_j \in \Nhop{i}{k}$ such that the following conditions hold:
\begin{enumerate}
    \item[(i)] $\Gamma_{\psi_i}(t) > 0$ for all $t\in\mathbb{R}_{\geq 0}$;
    \item[(ii)] $\rho^{\max}_{\psi_i} - \max_{\tau \in T_{\phi_i}^{\text{wdw}}} \rho_{\psi_i}^t(\tau) > 0$ for all $i \in \mathcal{V}$, where $T_{\phi_i}^{\text{wdw}}$, defined as in Section~\ref{Section: observer-based task satisfaction}, represents the time window over which $\phi_i$ must be satisfied;
    \item[(iii)]$ \hat{\rho}_\psi^{\text{opt}} := \text{sup}_{\vect{x}} \{\min_{i \in \mathcal{V}} [\rho^{\psi_i}(\vect{x}_{\phi_i})- \max_{\tau \in T_{\phi_i}^{\text{wdw}}} \rho_{\psi_i}^t(\tau)]\}$ satisfies $\hat{\rho}_\psi^{\text{opt}} > 0$.
\end{enumerate}
While condition (i) is required to properly define \eqref{eq: inequality on robustness to guarantee task satisfaction, observer not in design stage of gamma},
conditions (ii) and (iii) constitute a feasibility condition for the given $\rho_{\psi_i}^t(t)$ and $\rho^{\max}_{\psi_i}$.
From \eqref{eq: inequality on robustness to guarantee task satisfaction, observer not in design stage of gamma}, $\Gamma_{\psi_i}(t) = \gamma_{\psi_i}(t) - \rho_{\psi_i}^t(t)$, where both $\gamma_{\psi_i}(t)$ and $\delta_{N^i_j}^i(t)$ (for all $N^i_j \in \Neigh{i}{T}\setminus \Neighext{i}{C}$), on which $\rho_{\psi_i}^t(t)$ depends, are design parameters. Therefore, whenever Assumption~\ref{Assumption: assumption on task feasibility} holds, (i)-(iii) can be satisfied through an appropriate selection of $\delta_{N^i_j}^i(t)$, $\rho^{\max}_{\psi_i}$ and $\gamma_{\psi_i}(t)$ for all $i \in \mathcal{V}$ and $N^i_j \in \Neigh{i}{T}\setminus \Neighext{i}{C}$.

Under $\Gamma_{\psi_i}(t) > 0$, the following holds.
\begin{lemma}
    \label{lemma on prescribed performance function}
     $\Gamma_{\psi_i}(t)$ is a prescribed performance function as per Definition~\ref{Definition of Prescribed performance function}.
\end{lemma}
\begin{proof}
    The functions $\gamma_{\psi_i}(t)$ and $\rho_{\psi_i}^t(t)$ are prescribed performance functions as defined in \eqref{eq: robustness inequality no estimations ivolved} and Assumption~\ref{assumption on robustness lower bound}. Therefore: (i) Since $\gamma_{\psi_i}, \rho_{\psi_i}^t \in \mathcal{C}^1$, it follows that $\Gamma_{\psi_i} \in \mathcal{C}^1$; (ii) From Definition~\ref{Definition of Prescribed performance function}, there exist bounded constants $\overline{\gamma}_{\psi_i}, \overline{\rho}_{\psi_i}^t \in \mathbb{R}_{> 0}$ such that $0 < \gamma_{\psi_i} \leq \overline{\gamma}_{\psi_i}$ and $0< \rho_{\psi_i}^t < \overline{\rho}_{\psi_i}^t$. Then, if $\Gamma_{\psi_i}(t) > 0$ is satisfied by design, $0<\Gamma_{\psi_i}(t) < \overline{\gamma}_{\psi_i}$ holds; (iii) Since $|\dot{\gamma}_{\psi_i}(t)|\leq \overline{\gamma}_{d, \psi_i}$ and $|\dot{\rho}_{\psi_i}^t|\leq \overline{\rho}_{d,\psi_i}^t$ hold for some $\overline{\gamma}_{d, \psi_i} <\infty$ and $\overline{\rho}_{d,\psi_i}^t < \infty$, $|\dot{\Gamma}^{\psi_i}(t)| \leq |\dot{\gamma}_{\psi_i}(t)| + |\dot{\rho}_{\psi_i}^t|\leq \overline{\gamma}_{d, \psi_i} + \overline{\rho}_{d,\psi_i}^t$. As a result, $\Gamma_{\psi_i}(t)$ is a prescribed performance function.
\end{proof}

\begin{remark}
    \label{Remark: task modification in case of composite tasks}
    To handle tasks $\phi_i$ of the form \eqref{eq: definition of temporal formula}, where the \mbox{non-temporal} formula $\psi_i$ is obtained as the conjunction of $z_i$ subformulas, i.e., $\psi_i = \land_{j = 1}^{z_i} \psi_{i,j}$, we exploit the following properties of the always ($G$) and eventually ($F$) temporal operators:
\begin{equation}
    \label{eq: always operator properties}
    \begin{split}
        G_{[a,b]}\psi_i &= \land_{j = 1}^{z_i} G_{[a,b]} \psi_{i,j}, \\ F_{[a,b]}\psi_i &= G_{[\tau,\tau]} \psi_i \text{ with } \tau \in [a,b].
    \end{split}
\end{equation}
Given \eqref{eq: always operator properties}, each funnel function $\gamma_{\psi_{i,j}}(t)$ can be designed and adjusted independently for each subformula $\psi_{i,j}$, yielding the corresponding modified funnels  $\Gamma_{\psi_{i,j}}(t)$ for each $j \in \{1,\dots, z_i\}$. Satisfaction of $\phi_{i}$ is then enforced by imposing
\begin{equation}
    \label{eq: STL constraints in case of conjuction}
    - \bar{\Gamma}_{\psi_{i}}(t) < \bar{\rho}^{\psi_i}(\hat{\vect{x}}_{\phi_i}) - \rho^{\max}_{\psi_i} < 0
\end{equation}
as in \eqref{eq: inequality on robustness to guarantee task satisfaction, observer not in design stage of gamma}, with 
\begin{equation}
    \label{eq: smooth minimum among modified funnel}
    \bar{\Gamma}_{\psi_{i}}(t) = -\frac{1}{\eta} \ln(\sum_{j = 1}^{z_i}\exp(-\eta \Gamma_{\psi_{i,j}}(t))) 
\end{equation}
and
\begin{equation}
    \label{eq: smooth minimum among modified robustness}
    \bar{\rho}^{\psi_i}(\hat{\vect{x}}_{\phi_i}) = -\frac{1}{\eta} \ln(\sum_{j = 1}^{z_i}\exp(-\eta \rho^{\psi_{i,j}}(\hat{\vect{x}}_{\phi_i}))).
\end{equation}
Indeed, since $\bar{\rho}^{\psi_i}(\hat{\vect{x}}_{\phi_i})$ and $\bar{\Gamma}_{\psi_{i}}(t)$ in \eqref{eq: smooth minimum among modified funnel} and \eqref{eq: smooth minimum among modified robustness} satisfy $\bar{\Gamma}_{\psi_{i}}(t)\leq \min_{j \in \{1, \dots, z_i\}} \Gamma_{\psi_{i,j}}(t)$ and $\bar{\rho}^{\psi_i}(\hat{\vect{x}}_{\phi_i}) \leq \min_{j \in \{1, \dots, z_i\}} \rho^{\psi_{i,j}}(\hat{\vect{x}}_{\phi_i})$, \eqref{eq: STL constraints in case of conjuction} suffices to guarantee satisfaction of $\phi_{i}$.
\end{remark}

%% file: Content/Controller.tex
Building on \cite{LINDEMANN2021100973} and \cite{10918825}, this section adopts the Prescribed Performance Control (PPC) \cite{4639441} framework to address Problem \ref{Problem: first problem formulation}. Unlike existing approaches, we propose a novel observer-based controller that ensures task satisfaction by exploiting the observers’ prescribed-performance guarantees and the modified robustness functions introduced in Section~\ref{Sec: Observer-based Task Satisfaction}.

Under the condition that $\Gamma_{\psi_i}(t) > 0$ by design, the normalized error $e^{\psi_i}(\hat{\vect{x}}_{\phi_i}, t)$ is defined as $e^{\psi_i}(\hat{\vect{x}}_{\phi_i}, t) := \Gamma_{\psi_i}^{-1}(t) (\rho^{\psi_i}(\hat{\vect{x}}_{\phi_i}) - \rho^{\max}_{\psi_i})$.
The transformed error $\epsilon^{\psi_i}(\hat{\vect{x}}_{\phi_i}, t)$ is then given by
\begin{equation}
    \label{eq: transformed erro definition}
    \begin{split}
    \epsilon^{\psi_i}(\hat{\vect{x}}_{\phi_i}, t) &:= T_{\psi_i}(e^{\psi_i}(\hat{\vect{x}}_{\phi_i}, t)),
    \end{split}
\end{equation}
where $T_{\psi_i}: (-1,0) \rightarrow \mathbb{R}$ is a strictly increasing transformation $T_{\psi_i}(e^{\psi_i}(\hat{\vect{x}}_{\phi_i}, t)) := \ln\left(-\frac{e^{\psi_i}(\hat{\vect{x}}_{\phi_i}, t)+1}{e^{\psi_i}(\hat{\vect{x}}_{\phi_i}, t)}\right)$ with Jacobian $J_{\psi_i}(e^{\psi_i})= \frac{\partial{T_{\psi_i}}}{\partial e^{\psi_i}} = - \frac{1}{e^{\psi_i}(e^{\psi_i}+1)}$.

The central objective of PPC is to construct control laws $u_i$ that ensure the boundedness of $\epsilon^{\psi_i}(\hat{\vect{x}}_{\phi_i}, t)$, which in turn guarantees the satisfaction of \eqref{eq: inequality on robustness to guarantee task satisfaction, observer not in design stage of gamma}. To this end, consider the following assumption on $\rho^{\psi_i}$.

\begin{assumption}
    \label{Assumption on concavity and weelposeness of the robustness function}
    Each formula $\psi_i$ is such that: (i) $\rho^{\psi_i}: \mathbb{R}^{N_i^T} \rightarrow \mathbb{R}$ is concave with  $\frac{\partial \rho^{\psi_i}(\vect{x}_{\phi_i})}{\partial x_{i}} = 0_n$ only at the global maximum (ii) $\rho^{\psi_i}$ is continuously differentiable with respect to $\vect{x}_{\phi_i}$ or smooth almost everywhere except from the global maximum and (iii) the formula is \mbox{well-posed} in the sense that for all $C \in \mathbb{R}$ there exists $0 \leq \Bar{C} < \infty$ such that for all $\vect{x}_{\phi_i}$ with $\rho^{\psi_i}(\vect{x}_{\phi_i}) \geq C$, $\norm{\vect{x}_{\phi_i}}\leq \Bar{C}$ holds.
\end{assumption}

\begin{remark}
    \label{Remark: weel-posenes of robustness}
    Since linear functions and those modeling containment or formation tasks are concave, condition (i) of Assumption~\ref{Assumption on concavity and weelposeness of the robustness function} is not restrictive.
    Moreover, each $\phi_i$ can be augmented with $\psi^{\text{Bound}}_i := (\norm{\vect{x}_{\phi_i}}^2 \leq \Bar{C}_i^2)$; thus, condition (iii) of Assumption~\ref{Assumption on concavity and weelposeness of the robustness function} can be satisfied by selecting $\Bar{C}_i$ sufficiently large to preserve feasibility, provided that a suitable task assignment and observer convergence are ensured. For \mbox{collaborative} tasks, only $\Bar{\vect{x}}_{\phi_i}$ can be directly bounded via $\psi^{\text{Bound}}_i$. Nevertheless, as will be shown in Theorem~\ref{Theorem on task satisfacion}, under Assumption~\ref{Assumption on Task graph} and a convergent observer, it suffices to augment each $\psi_j$ that does not satisfy condition (iii) with $\psi^{\text{Bound}}_j := (\norm{\Bar{\vect{x}}_{\phi_j}}^2 \leq \Bar{C}_j^2)$ to guarantee boundedness of $\norm{\hat{\vect{x}}_{\phi_i}}$, and consequently ensure that (iii) also holds for $\rho^{\psi_i}(\hat{\vect{x}}_{\phi_i})$.
\end{remark}

Let $\rho^{\text{opt}}_{\psi_i} := \text{sup}_{\vect{x}_{\phi_i}} \rho^{\psi_i}(\vect{x}_{\phi_i})$ denote the local optimum. Given the concavity of $\rho^{\psi_i}$ from Assumption~\ref{Assumption on concavity and weelposeness of the robustness function}, to avoid $\frac{\partial\rho^{\psi_i}(\vect{x}_{\phi_i})}{\partial x_{i}} = 0_n$ along the transient, and thus feasibility issues, it suffices to tune $\rho^{\max}_{\psi_i}$ so that $0 < \rho^{\max}_{\psi_i} <\rho^{\text{opt}}_{\psi_i}$.

To present our main result, an additional assumption is required.
\begin{assumption}
    \label{Assumption on initialization}
    $-\Gamma_{\psi_i}(0) < {\rho}^{\psi_i}(\hat{\vect{x}}_{\phi_i}(0)) - \rho^{\max}_{\psi_i} < 0$ and $|\xi^{i}_{N^i_j}(0)| < \rho_{N_j^i}^i(0)$ hold for all $i \in \mathcal{V}$ and $N_j^i \in \Nhop{i}{k}$.
\end{assumption}

Assumption~\ref{Assumption on initialization} involves only design parameters and initialization quantities. Therefore, it's not restrictive.

With the notation and assumptions above, we now present the decentralized controller. 
\begin{theorem}
    \label{Theorem on task satisfacion}
    Consider a MAS \eqref{eq: agent's dynamic} with communication graph $\mathcal{G}_C$, and task graph $\mathcal{G}_T$ induced by a global STL task $\phi =\land_{i=1}^N \phi_i$. Suppose each agent runs the $k$-hop PPSO in \eqref{Eq: State observer dynamics} with $k$ selected according to \eqref{selection of the parameter k}. Consider the agents belonging to the cluster $\mathcal{C}_l$, subject to the task $\phi^{c}_l = \land_{l_i \in \mathcal{V}_l} \phi_{l_i}$. Then, $\phi^c_l$ is satisfied, for all $l \in \mathcal{L}$, if Assumptions \ref{Assumption on existence of a solution}-\ref{Assumption on Task graph}, \ref{Assumption on the non-existance of a task not in communication inside a cluster}-\ref{Assumption on initialization} hold and each agent $l_i\in \mathcal{V}_l$ applies the decentralized controller
    \begin{equation}
        \label{eq: agent's input}
        u_{l_i} = - g^{\top}_{l_i}(x_{l_i}(t))\sum_{l_j\in \mathcal{V}_l} \frac{\partial {\rho}^{\psi_{l_j}}(\hat{\vect{x}}_{\phi_{l_j}})}{\partial x_{{l_i}}} \Gamma^{-1}_{\psi_{l_j}}  J_{\psi_{l_j}} \epsilon^{\psi_{l_j}},
    \end{equation}
    where $\Gamma_{\psi_{l_j}}$ is defined as in \eqref{eq: inequality on robustness to guarantee task satisfaction, observer not in design stage of gamma}, and $J_{\psi_{l_j}}$, $\epsilon^{\psi_{l_j}}$ as in \eqref{eq: transformed erro definition}.
\end{theorem} 

\begin{proof}
    The proof is structured in three steps. In Step~A, the transformed error dynamics are derived for a generic cluster $\mathcal{C}_l$. In Step~B, task satisfaction is established for clusters $\mathcal{C}_l \in \mathcal{G}'$ corresponding to leaf clusters of $\mathcal{G}'$. Finally, in Step~C, observer convergence together with Assumption~\ref{Assumption on Task graph} is used to prove task satisfaction for all remaining clusters.
    
    For notational simplicity, $\hat{\rho}^{\psi_{l_i}} := \rho^{\psi_{l_i}}(\hat{\vect{x}}_{\phi_{l_i}})$ and $\rho^{\psi_{l_i}} := \rho^{\psi_{l_i}}(\vect{x}_{\phi_{l_i}})$ in the following.
    
    \textbf{Step A (transformed error dynamics):} 
    Given the definition in \eqref{eq: transformed erro definition}, $\dot{\epsilon}^{\psi_{l_i}}(\hat{\vect{x}}_{\phi_{l_i}}, t) = -\frac{1}{e^{\psi_{l_i}}(1+e^{\psi_{l_i}})}\dot{e}^{\psi_{l_i}}$ holds with $\dot{e}^{\psi_{l_i}} =  \Gamma_{\psi_{l_i}}^{-1}\Bigl(\frac{\partial \hat{\rho}^{\psi_{l_i}}}{\partial \Bar{\vect{x}}_{\phi_{l_i}}}^{\top} \dot{\Bar{\vect{x}}}_{\phi_{l_i}} + \frac{\partial \hat{\rho}^{\psi_{l_i}}}{\partial \estimate{\vect{x}}{{l_i}}{{\phi_{l_i}}}}^{\top}  \dot{\Hat{\vect{x}}}^{l_i}_{\phi_{l_i}} - \dot{\Gamma}_{\psi_{l_i}}(t)e^{\psi_{l_i}}\Bigr)$, where $\frac{\partial \hat{\rho}^{\psi_{l_i}}}{\partial \Bar{\vect{x}}_{\phi_{l_i}}} \in \mathbb{R}^{N^{TC}_{{l_i}}}$ and $\frac{\partial \hat{\rho}^{\psi_{l_i}}}{\partial \estimate{\vect{x}}{{l_i}}{{\phi_{l_i}}}} \in \mathbb{R}^{N^{T\setminus C}_{{l_i}}}$ are vectors containing, respectively, the derivatives of $\hat{\rho}^{\psi_{l_i}}$ with respect to the state of agents $j\in \Neigh{{l_i}}{T} \cap \Neighext{{l_i}}{C}$ and $j\in \Neigh{{l_i}}{T} \setminus \Neighext{{l_i}}{C}$, and where $\dot{\Bar{\vect{x}}}_{\phi_{l_i}}$ and $\dot{\Hat{\vect{x}}}^{l_i}_{\phi_{l_i}}$ denote respectively the stacked dynamics of $x_j$ for $j\in \Neigh{{l_i}}{T} \cap \Neighext{{l_i}}{C}$, and the one of $\hat{x}^{l_i}_j$, $j\in \Neigh{{l_i}}{T} \setminus \Neighext{{l_i}}{C}$, when the estimate is performed by $i$.  
    As introduced in Section \ref{Section: collaborative and non-collaborative tasks}, $\Bar{\vect{x}}_{\phi_{l_i}}$ is the vector containing the true state of agents  $j\in \Neigh{{l_i}}{T} \cap \Neighext{{l_i}}{C}$. Thus, $\dot{\epsilon}^{\psi_{l_i}}$ can be rewritten as:
    \begin{equation}
        \label{eq: transformed normalized error dynamics}
        \begin{split}
            \dot{\epsilon}^{\psi_{l_i}} = & J_{\psi_{l_i}}\Gamma^{-1}_{\psi_{l_i}}\biggl(\frac{\partial \hat{\rho}^{\psi_{l_i}}}{\partial {\vect{x}}_{\phi_{l_i}}}^{\top}  \dot{{\vect{x}}}_{\phi_{l_i}}  + \frac{\partial \hat{\rho}^{\psi_{l_i}}}{\partial \estimate{\vect{x}}{{l_i}}{{\phi_{l_i}}}}^{\top}  \dot{\Hat{\vect{x}}}^{i}_{\phi_{l_i}}  -\dot{\Gamma}_{\psi_{l_i}}e^{\psi_{l_i}}\biggr),
        \end{split}
    \end{equation}
    where $\frac{\partial \hat{\rho}^{\psi_{l_i}}}{\partial {\vect{x}}_{\phi_{l_i}}}$ is the derivative of $\hat{\rho}^{\psi_{l_i}}$ with respect to the state of the agents in cluster $\mathcal{C}_l  = (\mathcal{V}_l, \mathcal{E}_l)$, with $l_i \in \mathcal{V}_l$. Note that, from condition (i) in Assumption \ref{Assumption on concavity and weelposeness of the robustness function}, and since $0 < \rho^{\max}_{\psi_j} <\rho^{\text{opt}}_{\psi_j}$ holds for all $j \in \mathcal{V}$, the only components of $\frac{\partial \hat{\rho}^{\psi_{l_i}}}{\partial {\vect{x}}_{\phi_{l_i}}}$ equal to zero are those associated to agents $r \notin \Neigh{{l_i}}{T}$ and $r \in \Neigh{{l_i}}{T} \setminus \Neighext{{l_i}}{C}$.
    By stacking the transformed errors of cluster $\mathcal{C}_l$, i.e., $\epsilon^{\psi_{l_i}}$ for all $l_i \in \mathcal{V}_l$, we define the transformed error vector  $\vect{\epsilon}^{\phi^c_l} := [\epsilon^{\psi_{l_1}}, \dots, \epsilon^{\psi_{l_{v_l}}}]^{\top}$  with dynamics
    \begin{equation}
        \hspace{-0.2cm}
        \label{eq: cluster transformed normalized error dynamics}
        \dot{ \vect{\epsilon}}^{\phi^c_l} = \vect{J}_{\phi^c_l} \vect{\Gamma}^{-1}_{\phi^c_l}  \{ \vect{\Lambda}_{\vect{x}_{\phi^c_l}}^{\top} \dot{\vect{x}}_{\phi^c_l} + \vect{\Lambda}_{\Hat{\vect{x}}_{\phi^c_l}}^{\top}\dot{\Hat{\vect{x}}}_{\phi^c_l} - \dot{\vect{\Gamma}}_{\phi^c_l} \vect{e}^{\phi^c_l}\},
    \end{equation}
    where $\vect{x}_{\phi^c_l} := [x^\top_{l_1}, \dots,  x^\top_{l_{v_l}}]^{\top}$ and $\Hat{\vect{x}}_{\phi^c_l} := [\estimate{\vect{x}}{{l_1 \top}}{{\phi_{l_1}}}, \dots , \estimate{\vect{x}}{{l_{v_l}} \top}{{\phi_{l_{v_l}}}}]^\top $ are defined as in Section \ref{Section: collaborative and non-collaborative tasks}, $\vect{J}_{\phi^c_l} := \text{diag}(J_{\psi_{l_1}}, \dots, J_{\psi_{l_{v_l}}} )$, $\vect{\Gamma}_{\phi^c_l} := \text{diag}(\Gamma_{\psi_{l_1}}, \dots, \Gamma_{\psi_{l_{v_l}}})$,  $\vect{\Lambda}_{\vect{x}_{\phi^c_l}} := \Bigl[ \frac{\partial \hat{\rho}^{\psi_{l_1}}}{\partial \vect{x}_{\phi^c_l}} \ \dots \  \frac{\partial \hat{\rho}^{\psi_{l_{v_l}}}}{\partial \vect{x}_{\phi^c_l}} \Bigr]$, 
      $\vect{\Lambda}_{\Hat{\vect{x}}_{\phi^c_l}} := \text{diag}\Bigl( \frac{\partial \hat{\rho}^{\psi_{l_1}}}{\partial \estimate{\vect{x}}{{l_1}}{{\phi_{l_1}}}}, \dots , \frac{\partial \hat{\rho}^{\psi_{l_{v_l}}}}{\partial \estimate{\vect{x}}{l_{v_l}}{{\phi_{l_{v_l}}}}}\Bigr)$ and $\vect{e}^{\phi^c_l} := [e^{\psi_{l_1}}, \dots,  e^{\psi_{l_{v_l}}}]^{\top}$.
    Given the definition of $\vect{x}_{\phi^c_l}$ from Section \ref{Section: collaborative and non-collaborative tasks}, its dynamics can be written as:
    \begin{equation}
        \label{eq: cluster state dynamcis}
        \dot{\vect{x}}_{\phi^c_l} = \vect{f}_{\phi^c_l}(\vect{x}_{\phi^c_l}) +\vect{g}_{\phi^c_l}(\vect{x}_{\phi^c_l})\vect{u}_{\phi^c_l} + \vect{w}_{\phi^c_l}(\vect{x},t),
    \end{equation}
    where $\vect{f}_{\phi^c_l}(\vect{x}_{\phi^c_l}) := [f_{l_1}(x_{l_1}), \dots ,f_{l_{v_l}}(x_{l_{v_l}})]^\top$, $\vect{g}_{\phi^c_l}(\vect{x}_{\phi^c_l}) := \text{diag}(g_{l_1}(x_{l_1}), \dots ,g_{l_{v_l}}(x_{l_{v_l}}))$,  $\vect{u}_{\phi^c_l} := [u_{l_1}, \dots, u_{l_{v_l}}]^\top$ and $\vect{w}_{\phi^c_l}(\vect{x},t) := [w_{l_1}(\vect{x},t), \dots, w_{l_{v_l}}(\vect{x},t)]^\top$.    
   
    \textbf{Step~B (leaf clusters):} As discussed in Remark \ref{remark on the acyclicity of the cluster induced graph}, $\mathcal{G}'$ is a directed acyclic graph. Therefore, there exists at least one leaf cluster $\mathcal{C}_l = (\mathcal{V}_l, \mathcal{E}_l)$, $l \in \mathcal{L}$, with no outgoing edges. Consequently, from Definition \ref{def: cluster induced graph} and Assumption~\ref{Assumption on the non-existance of a task not in communication inside a cluster}, all $\phi_{l_i}$, with $l_i\in \mathcal{V}_l$, involve only agents in $\mathcal{V}_l$, and satisfaction of
    $\phi_l^c$ is independent of estimations. As a result, $\hat{\rho}^{\psi_{l_i}} \equiv \rho^{\psi_{l_i}}$, $\frac{\partial \hat{\rho}^{\psi_{l_i}}}{\partial \estimate{\vect{x}}{{l_i}}{{\phi_{l_i}}}}^{\top} = 0$, and \eqref{eq: cluster transformed normalized error dynamics} boils down to $\dot{\vect{\epsilon}}^{\phi^c_l} = \vect{J}_{\phi^c_l} \vect{\Gamma}^{-1}_{\phi^c_l}  \{ \vect{\Lambda}_{\vect{x}_{\phi^c_l}}^{\top} \dot{\vect{x}}_{\phi^c_l}- \dot{\vect{\Gamma}}_{\phi^c_l} \vect{e}^{\phi^c_l}\}$, where $\vect{\Lambda}_{\vect{x}_{\phi^c_l}} := \Bigl[ \frac{\partial {\rho}^{\psi_{l_1}}}{\partial \vect{x}_{\phi^c_l}} \ \dots \  \frac{\partial {\rho}^{\psi_{l_{v_l}}}}{\partial \vect{x}_{\phi^c_l}} \Bigr]$ and $\vect{\Gamma}_{\phi^c_l} := \text{diag}(\gamma_{\psi_{l_1}}, \dots,\gamma_{\psi_{l_{v_l}}})$.
    
    Consider the candidate Lyapunov function $V(\vect{\epsilon}^{\phi^c_l} (\vect{e}^{\phi^c_l})) = \frac{1}{2} \vect{\epsilon}^{\phi^c_l} (\vect{e}^{\phi^c_l})^\top \vect{\epsilon}^{\phi^c_l} (\vect{e}^{\phi^c_l})$. By differentiating $V$ and by replacing $\dot{\vect{\epsilon}}^{\phi^c_l}$, $\dot{V}$ becomes $ \dot{V} =  \vect{\epsilon}^{\phi^c_l\top} \vect{J}_{\phi^c_l} \vect{\Gamma}^{-1}_{\phi^c_l}  \{ \vect{\Lambda}_{\vect{x}_{\phi^c_l}}^{\top} \dot{\vect{x}}_{\phi^c_l}- \dot{\vect{\Gamma}}_{\phi^c_l} \vect{e}^{\phi^c_l}\}$. Then, by plugging \eqref{eq: cluster state dynamcis}, $ \dot{V} =  \vect{\epsilon}^{\phi^c_l\top} \vect{J}_{\phi^c_l} \vect{\Gamma}^{-1}_{\phi^c_l}  \{ \vect{\Lambda}_{\vect{x}_{\phi^c_l}}^{\top} [\vect{f}_{\phi^c_l}(\vect{x}_{\phi^c_l}) +  \vect{w}_{\phi^c_l}(\vect{x},t) +\vect{g}_{\phi^c_l}(\vect{x}_{\phi^c_l})\vect{u}_{\phi^c_l} ]- \dot{\vect{\Gamma}}_{\phi^c_l} \vect{e}^{\phi^c_l}\}$ holds.
    By stacking $u_{l_i}$ as per \eqref{eq: agent's input} for all $l_i \in \mathcal{V}_l$, the input vector is $\vect{u}_{\phi^c_l} = -\vect{g}^{\top}_{\phi^c_l}(\vect{x}_{\phi^c_l})  \vect{\Lambda}_{\vect{x}_{\phi^c_l}}^{\top} \vect{\Gamma}^{-1}_{\phi^c_l} \vect{J}_{\phi^c_l}\vect{\epsilon}^{\phi^c_l}$, and 
     \begin{equation} 
        \label{eq: Lyapunov function expression}
        \begin{split}
            \dot{V} =& - \vect{\epsilon}^{\phi^c_l\top} \vect{J}_{\phi^c_l} \vect{\Gamma}^{-1}_{\phi^c_l} \vect{\Lambda}_{\vect{x}_{\phi^c_l}}^{\top} \Theta_{\phi^c_l}  \vect{\Lambda}_{\vect{x}_{\phi^c_l}} \vect{\Gamma}^{-1}_{\phi^c_l} \vect{J}_{\phi^c_l}\vect{\epsilon}^{\phi^c_l} \\  &+ \vect{\epsilon}^{\phi^c_l\top} \vect{J}_{\phi^c_l} \vect{\Gamma}^{-1}_{\phi^c_l} \{- \dot{\vect{\Gamma}}_{\phi^c_l} \vect{e}^{\phi^c_l} + \vect{\Lambda}_{\vect{x}_{\phi^c_l}}^{\top} [\vect{f}_{\phi^c_l}+ \vect{w}_{\phi^c_l}]\},
        \end{split}
    \end{equation}
    where $\Theta_{\phi^c_l} := \vect{g}_{\phi^c_l}(\vect{x}_{\phi^c_l}) \vect{g}^\top_{\phi^c_l}(\vect{x}_{\phi^c_l})$ is a positive definite block diagonal matrix by condition (ii) of Assumption~\ref{Assumption on existence of a solution}.
    
    Given that $\mathcal{G}_T$ is acyclic from Assumption~\ref{Assumption on Task graph}, there exists a topological reordering of its nodes, i.e., an ordering such that each directed edge $(i,j) \in \mathcal{G}_T$ implies $j < i$. Accordingly, there exist permutation matrices $P_{\vect{x}_{\phi^c_l}}$ and $P_{\phi^c_l}$ such that $P_{\vect{x}_{\phi^c_l}} \vect{\Lambda}_{\vect{x}_{\phi^c_l}}P^\top_{\phi^c_l}$ admits a block \mbox{lower-triangular} structure induced by the task dependencies. This reflects that each task depends only on state components associated with \mbox{lower-indexed} nodes in the ordering, while no dependence exists on \mbox{higher-indexed} nodes. Since permutation matrices are invertible and induce only relabeling of coordinates, they preserve rank and positive definiteness of quadratic forms under congruence transformations. Hence, without loss of generality, we assume that the agents and task indices are ordered according to a topological ordering, and we omit the corresponding permutation matrices in the following.
    
    Since $\rho^{\max}_{\psi_{l_i }} <\rho^{\text{opt}}_{\psi_{l_i }}$ holds by design, and (i) of Assumption~\ref{Assumption on concavity and weelposeness of the robustness function} guarantees that $\frac{\partial {\rho}^{\psi_{l_i}}}{\partial x_{l_i}} \neq 0$ for all $l_i \in \mathcal{V}_l$ whenever ${\rho}^{\psi_{l_i}} < \rho^{\text{opt}}_{\psi_{l_i }}$, the elements on the main diagonal of $\vect{\Lambda}_{\vect{x}_{\phi^c_l}}$ are nonzero. 
    As a result, since $\vect{\Lambda}_{\vect{x}_{\phi^c_l}}$ is block lower-triangular with nonzero elements on the main diagonal, $\text{rank}(\vect{\Lambda}_{\vect{x}_{\phi^c_l}}) = v_l$ and $\vect{\Lambda}_{\vect{x}_{\phi^c_l}}^{\top} \Theta_{\phi^c_l} \vect{\Lambda}_{\vect{x}_{\phi^c_l}}$ is positive definite.
    Since $\vect{J}_{\phi^c_l}$ and $\vect{\Gamma}^{-1}_{\phi^c_l} $ are diagonal matrices with positive diagonal entries $J_{\psi_{l_i}}(e^{\psi_{l_i}})$ and $\gamma_{\psi_{l_i}}(t)$, $- \vect{\epsilon}^{\phi^c_l\top } \vect{J}_{\phi^c_l} \vect{\Gamma}^{-1}_{\phi^c_l}\vect{\Lambda}_{\vect{x}_{\phi^c_l}}^{\top} \Theta_{\phi^c_l} \vect{\Lambda}_{\vect{x}_{\phi^c_l}} \vect{\Gamma}^{-1}_{\phi^c_l} \vect{J}_{\phi^c_l}\vect{\epsilon}^{\phi^c_l} <  - \alpha_{\rho} \alpha_{J} \vect{\epsilon}^{\phi^c_l \top}\vect{\epsilon}^{\phi^c_l} $ holds with $\alpha_{\rho} = \lambda_{\text{min}}(\vect{\Lambda}_{\vect{x}_{\phi^c_l}}^{\top} \Theta_{\phi^c_l} \vect{\Lambda}_{\vect{x}_{\phi^c_l}}) \in \mathbb{R}_{> 0}$ and $\alpha_J = \min_{l_i \in \mathcal{V}_l} \Bigl\{\min_{t \in \mathbb{R}_{\geq 0},e^{\psi_{l_i}}\in (-1,0) } \gamma_{\psi_{l_i}}^{-2}(t) J^2_{\psi_{l_i}}(e^{\psi_{l_i}}) \Bigr\} \in \mathbb{R}_{> 0}$. By adding and subtracting $\zeta \norm{\vect{\Gamma}^{-1}_{\phi^c_l} \vect{J}_{\phi^c_l}\vect{\epsilon}^{\phi^c_l}}^2$ for some $0 < \zeta < \alpha_\rho$, \eqref{eq: Lyapunov function expression} can be rewritten as
     \begin{equation}
        \label{eq: rewritten Lyapunov function bound, step B proof task satisfaction}
        \begin{split}
             \dot{V} \leq &- (\alpha_{\rho} - \zeta) \alpha_J \norm{\vect{\epsilon}^{\phi^c_l}}^2 + \vect{\epsilon}^{\phi^c_l \top} \vect{J}_{\phi^c_l} \vect{\Gamma}^{-1}_{\phi^c_l} b(t) \\ &- \zeta \norm{\vect{\Gamma}^{-1}_{\phi^c_l} \vect{J}_{\phi^c_l}\vect{\epsilon}^{\phi^c_l}}^2,
        \end{split}
     \end{equation}
     where $b(t) :=  \Bigl\{- \dot{\vect{\Gamma}}_{\phi^c_l} \vect{e}^{\phi^c_l} + \vect{\Lambda}_{\vect{x}_{\phi^c_l}}^{\top} [\vect{f}_{\phi^c_l}+ \vect{w}_{\phi^c_l}]\Bigr\}$. Since $\vect{\epsilon}^{\phi^c_l \top} \vect{J}_{\phi^c_l} \vect{\Gamma}^{-1}_{\phi^c_l} b(t) - \zeta \norm{\vect{\Gamma}^{-1}_{\phi^c_l} \vect{J}_{\phi^c_l}\vect{\epsilon}^{\phi^c_l}}^2$ resemble terms of the quadratic form $\norm{\sqrt{\zeta} \vect{\Gamma}^{-1}_{\phi^c_l} \vect{J}_{\phi^c_l}\vect{\epsilon}^{\phi^c_l} - \frac{1}{2\sqrt{\zeta}}b(t)}^2$, and $\norm{\sqrt{\zeta} \vect{\Gamma}^{-1}_{\phi^c_l} \vect{J}_{\phi^c_l}\vect{\epsilon}^{\phi^c_l} - \frac{1}{2\sqrt{\zeta}}b(t)}^2 \geq 0$ hold by the definition of norm, the inequality $\vect{\epsilon}^{\phi^c_l \top} \vect{J}_{\phi^c_l} \vect{\Gamma}^{-1}_{\phi^c_l} b(t) - \zeta \norm{\vect{\Gamma}^{-1}_{\phi^c_l} \vect{J}_{\phi^c_l}\vect{\epsilon}^{\phi^c_l}}^2 \leq \frac{1}{4\zeta}b^\top(t)b(t)$ is satisfied for all $\vect{\epsilon}^{\phi^c_l}$ and \eqref{eq: rewritten Lyapunov function bound, step B proof task satisfaction} can be bounded as $ \dot{V} \leq - (\alpha_{\rho} - \zeta) \alpha_J \norm{\vect{\epsilon}^{\phi^c_l}}^2 + \frac{1}{4\zeta}b^\top(t)b(t)$. As a result:
     \begin{equation}
        \label{eq: Lyapunov function final upper bound}
        \dot{V} \leq - \kappa V + \vect{b}(t),
    \end{equation}
    with $\kappa := 2(\alpha_{\rho} - \zeta) \alpha_J $ and $\vect{b}(t) := \frac{1}{4\zeta} (\norm{\vect{\Lambda}_{\vect{x}_{\phi^c_l}}}\norm{\vect{f}_{\phi^c_l}} + \norm{\vect{\Lambda}_{\vect{x}_{\phi^c_l}}}\norm{\vect{w}_{\phi^c_l}} + \norm{\dot{\vect{\Gamma}}_{\phi^c_l}  \vect{e}^{\phi^c_l}})^2$. To proceed, let's check whether $\vect{b}(t)$ admits an upper bound $\Bar{\vect{b}}(t)$. For this purpose, define the set $\mathcal{X}_{l_i}(t):=\{\vect{x}_{\phi_{l_i}} \in \mathbb{R}^{N^{TC}_{i}}| -1 < e^{\psi_{l_i}} < 0 \}$ as the one containing the states $\vect{x}_{\phi_{l_i}}$ satisfying task $\phi_{l_i}$ at time $t$. Since $\gamma_{\phi_{l_i}}(t)$ is a non-increasing function, $\mathcal{X}_{l_i}(t_2) \subseteq \mathcal{X}_{l_i}(t_1)$ holds for all $t_1 < t_2$, and $\mathcal{X}_{l_i}(0)$ collects all states $\vect{x}_{\phi_{l_i}}$ for which $e^{\psi_{l_i}} \in (-1,0)$ for all $t \in \mathbb{R}_{\geq 0}$. From condition (iii) of Assumption~\ref{Assumption on concavity and weelposeness of the robustness function}, $\mathcal{X}_{l_i}(0)$ is bounded for all $l_i \in \mathcal{V}_l$. Thus, since $f$ is Lipschitz continuous, $\norm{\vect{f}_{\phi^c_l}(\vect{x}_{\phi^c_l})}$ is bounded in $\mathcal{X}_{l_i}(0)$. Moreover, since $\dot{\vect{\Gamma}}_{\phi^c_l} \vect{e}^{\phi^c_l}$ is a column vector with $\dot{\gamma}_{\phi_{l_i}}e^{\psi_{l_i}}$ as entries, and $|\dot{\gamma}_{\phi_{l_i}}(t)e^{\psi_{l_i}}| < |\dot{\gamma}_{\phi_{l_i}}(0)|$ holds with $|\dot{\gamma}_{\phi_{l_i}}(0)|< \infty$ from $\gamma_{\phi_{l_i}}$ definition, also $\norm{\dot{\gamma}_{\phi_{l_i}}e^{\psi_{l_i}}}$ is bounded. 
    Note that ${\rho}^{\psi_{l_i}}$ is concave, continuously differentiable and \mbox{well-posed} according to Assumption~\ref{Assumption on concavity and weelposeness of the robustness function}. Thus, since every function of this kind has bounded derivative on bounded open sets contained in its domain, $\norm{\vect{\Lambda}_{\vect{x}_{\phi^c_l}}}$ is bounded on $\mathcal{X}_{l_i}(0)$.
    Finally, since $w_i(\vect{x},t)$ satisfies Assumption \ref{Assumption on existence of a solution}, $\norm{\vect{w}_{\phi^c_l}}$ is bounded and an upper bound $\Bar{\vect{b}}(t)$ on $\vect{b}(t)$ is guaranteed to exist for all $\vect{x}_{\phi_{l_i}} \in \mathcal{X}_{l_i}(0)$ and all $l_i \in \mathcal{V}_l$. 
    
    To conclude Step~A, consider $S(\vect{e}^{\phi^c_l}) := 1 - \exp(- V(\vect{\epsilon}^{\phi^c_l} (\vect{e}^{\phi^c_l})))$. From its definition: (i) $S(\vect{e}^{\phi^c_l}) \in (0,1)$ for all $\vect{e}^{\psi_{l_i}} \in \mathcal{D} $, with $\mathcal{D} = \bigtimes_{j=1}^{v_l}(-1,0)$, and (ii) $S(\vect{e}^{\phi^c_l}) \rightarrow 1$ as $\vect{e}^{\phi^c_l} \rightarrow \partial \mathcal{D} $. Thus, studying the boundedness of $\vect{\epsilon}^{\phi^c_l}$ through the one of $V$, reduces to proving that $S(\vect{e}^{\phi^c_l}) < 1$ holds for all $t$. By differentiating $S(\vect{e}^{\phi^c_l})$, and by replacing \eqref{eq: Lyapunov function final upper bound} and $V(\vect{\epsilon}^{\phi^c_l} (\vect{e}^{\phi^c_l})) = - \ln (1-S(\vect{e}^{\phi^c_l}))$, the following inequality is obtained: $\dot{S}(\vect{e}^{\phi^c_l}) \leq - \kappa (1-S(\vect{e}^{\phi^c_l})) \Bigl(-\frac{1}{\kappa}\vect{b}(t) - \ln (1-S(\vect{e}^{\phi^c_l}))\Bigr)$.
    Since $\kappa$ and $1-S(\vect{e}^{\phi^c_l})$ are positive, they do not affect the sign of $ \dot{S}(\vect{e}^{\phi^c_l})$. Therefore, to verify whether $\dot{S}(\vect{e}^{\phi^c_l}) \leq 0 $, it suffices to study under which condition $-\frac{1}{\kappa}\vect{b}(t) - \ln (1-S(\vect{e}^{\phi^c_l}))\geq 0$ holds. Note that
 $-\frac{1}{\kappa}\vect{b}(t) - \ln (1-S(\vect{e}^{\phi^c_l}))\geq 0 $ is satisfied for all $\vect{e}^{\phi^c_l} \in \Omega^c_{\vect{e}}$, where $\Omega_{\vect{e}} :=\bigl\{\vect{e}^{\phi^c_l} \in \mathcal{D}| S(\vect{e}^{\phi^c_l}) < 1- \exp{(-\frac{\bar{\vect{b}}(t)}{\kappa})} \bigr\}$. Furthermore, $-\frac{1}{\kappa}\vect{b}(t) - \ln (1-S(\vect{e}^{\phi^c_l})) = 0 $ holds for $\vect{e}^{\phi^c_l} \in  \partial\Omega_{\vect{e}}$. Consequently, $\dot{S}(\vect{e}^{\phi^c_l}) \leq 0$  for $\vect{e}^{\phi^c_l} \in \Omega^c_{\vect{e}}$, with $\dot{S}(\vect{e}^{\phi^c_l}) =0$ iff $\vect{e}^{\phi^c_l} \in \partial\Omega_{\vect{e}}$. 
    Since the initialization satisfies $\rho^{\max}_{\psi_{l_i}} -\gamma_{\psi_{l_i}}(0) < \rho^{\psi_{l_i}}(\Bar{\vect{x}}_{\phi_{l_i}}(0)) < \rho^{\max}_{\psi_{l_i}}$ for all $l_i \in \mathcal{V}_l$, it follows that $e^{\psi_{l_i}} \in (-1,0)$ and $S(\vect{e}^{\phi^c_l}(0)) < 1$. Moreover, since $\exp{(-\frac{\Bar{\vect{b}}(t)}{\kappa})} > 0$ by definition, $S(\vect{e}^{\phi^c_l})) < 1 $ is preserved for all $t$, independently of whether $\vect{e}^{\phi^c_l}$ is initialized inside or outside $\Omega_{\vect{e}}$. From $S(\vect{e}^{\phi^c_l})) < 1 $, boundedness of $V(\vect{\epsilon}^{\phi^c_l} (\vect{e}^{\phi^c_l}))$, and therefore of the transformed error $\vect{\epsilon}^{\phi^c_l}$, follows. As a result, ${\rho}^{\psi_{l_i}}$ satisfies \eqref{eq: inequality on robustness to guarantee task satisfaction, observer not in design stage of gamma} for all $l_i \in \mathcal{V}_l$, and tasks satisfaction is guaranteed for every leaf cluster. 
    
    \textbf{Step C (non-leaf clusters):} Consider $\mathcal{C}_l$ that are not leaf clusters. For simplicity, focus on an \mbox{in-neighbor} of a leaf cluster, i.e., $\mathcal{C}_l$ such that $(\mathcal{C}_l, \mathcal{C}_j) \in \mathcal{E}'$ with $\mathcal{C}_j$ being a leaf cluster. From Step~B, $\psi_{j_q}$ is satisfied for all  $j_q \in \mathcal{V}_j$. Hence, $\epsilon^{\psi_{j_q}} < \infty$ for all $j_q \in \mathcal{V}_j$, and $u_{j_q} < \infty$ holds. Since $\vect{x}_{\phi_{j_q}} \in \mathcal{X}_{j_q}(t)$, with $\mathcal{X}_{j_q}(t)$ bounded, all trajectories of agents  $j_q \in \mathcal{V}_j$ are bounded. Thus, Assumption \ref{Assumption: observer asssumptions for convergence} holds, Theorem~\ref{Theorem: main theorem on state estimation convergence} applies, and the \mbox{$k$-hop} PPSO in \eqref{Eq: State observer dynamics} guarantees $|\tilde{x}^{N^{j_q}_r}_{j_q}| < \delta^{N^{j_q}_r}_{j_q}$ for all $j_q\in \mathcal{V}_j$ and $N^{j_q}_r \in \Nhop{j_q}{k}$.
    Based on the observer guarantees, task satisfaction of non-leaf clusters is proven similarly to Step~B.
    Consider $V(\vect{\epsilon}^{\phi^c_l} (\vect{e}^{\phi^c_l})) = \frac{1}{2} \vect{\epsilon}^{\phi^c_l} (\vect{e}^{\phi^c_l})^\top \vect{\epsilon}^{\phi^c_l} (\vect{e}^{\phi^c_l})$. By means of \eqref{eq: agent's input}, \eqref{eq: cluster transformed normalized error dynamics} and \eqref{eq: cluster state dynamcis}, $\dot{V}$ can be written as
    \begin{equation} 
        \label{eq: Lyapunov function expression in inner clusters}
        \begin{split}
        \hspace{-0.1cm}
            \dot{V} =& \vect{\epsilon}^{\phi^c_l\top} \vect{J}_{\phi^c_l} \vect{\Gamma}^{-1}_{\phi^c_l} \{- \dot{\vect{\Gamma}}_{\phi^c_l} \vect{e}^{\phi^c_l} + \vect{\Lambda}_{\Hat{\vect{x}}_{\phi^c_l}}^{\top}\dot{\Hat{\vect{x}}}_{\phi^c_l} + \vect{\Lambda}_{\vect{x}_{\phi^c_l}}^{\top} [\vect{f}_{\phi^c_l}+\\  & \vect{w}_{\phi^c_l}] \} - \vect{\epsilon}^{\phi^c_l\top} \vect{J}_{\phi^c_l} \vect{\Gamma}^{-1}_{\phi^c_l} \vect{\Lambda}_{\vect{x}_{\phi^c_l}}^{\top} \Theta_{\phi^c_l}  \vect{\Lambda}_{\vect{x}_{\phi^c_l}} \vect{\Gamma}^{-1}_{\phi^c_l} \vect{J}_{\phi^c_l}\vect{\epsilon}^{\phi^c_l},
        \end{split}
    \end{equation}
    where $\vect{J}_{\phi^c_l}$, $\vect{\Gamma}^{-1}_{\phi^c_l}$, $\vect{\Lambda}_{\Hat{\vect{x}}_{\phi^c_l}}$ and $\vect{\Lambda}_{\vect{x}_{\phi^c_l}}$ are defined as in \eqref{eq: cluster transformed normalized error dynamics}, and $\Theta_{\phi^c_l}$ as in \eqref{eq: Lyapunov function expression}. 
    As in Step~A, under Assumption~\ref{Assumption on Task graph}, $\vect{\Lambda}_{\vect{x}_{\phi^c_l}}$ can be written as a block \mbox{lower-triangular} matrix with $\frac{\partial \hat{\rho}^{\psi_{l_i}}}{\partial x_{l_i}}$, $l_i \in \mathcal{V}_l$, on the main diagonal.
    
    Then, since $\rho^{\max}_{\psi_{l_i }} <\rho^{\text{opt}}_{\psi_{l_i }}$ by design, and every task depends on the true state of the agents to which it is assigned, $\text{rank}(\vect{\Lambda}_{\vect{x}_{\phi^c_l}}) = v_l$ and $\vect{\Lambda}_{\vect{x}_{\phi^c_l}}^{\top} \Theta_{\phi^c_l} \vect{\Lambda}_{\vect{x}_{\phi^c_l}}$ is positive definite. Since $J_{\psi_{l_i}}(e^{\psi_{l_i}}) >0$ and $\Gamma_{\psi_{l_i}}(t)>0$, $- \vect{\epsilon}^{\phi^c_l\top} \vect{J}_{\phi^c_l} \vect{\Gamma}^{-1}_{\phi^c_l} \vect{\Lambda}_{\vect{x}_{\phi^c_l}}^{\top} \Theta_{\phi^c_l}  \vect{\Lambda}_{\vect{x}_{\phi^c_l}} \vect{\Gamma}^{-1}_{\phi^c_l} \vect{J}_{\phi^c_l}\vect{\epsilon}^{\phi^c_l} <  - \alpha_{\rho} \alpha_J \vect{\epsilon}^{\phi^c_l \top}\vect{\epsilon}^{\phi^c_l}$ holds with $\alpha_{\rho} = \lambda_{\text{min}}(\vect{\Lambda}_{\vect{x}_{\phi^c_l}}^{\top} \Theta_{\phi^c_l} \vect{\Lambda}_{\vect{x}_{\phi^c_l}}) \in \mathbb{R}_{> 0}$ and $\alpha_J = \min_{l_i \in \mathcal{V}_l} \Bigl\{\min_{t \in \mathbb{R}_{\geq 0}, e^{\psi_{l_i}}\in (-1,0)} \Gamma_{\psi_{l_i}}^{-2}(t) J^2_{\psi_{l_i}}(e^{\psi_{l_i}}) \Bigr\} \in \mathbb{R}_{> 0}$.
    
    By adding and subtracting $\zeta \norm{\vect{\Gamma}^{-1}_{\phi^c_l} \vect{J}_{\phi^c_l}\vect{\epsilon}^{\phi^c_l}}^2$ for some $0 < \zeta < \alpha_\rho$, \eqref{eq: Lyapunov function expression in inner clusters} can be rewritten as:
    \begin{equation}
        \label{eq: Lyapunov function final upper bound, observation case}
        \dot{V} \leq - \kappa V + \vect{b}(t),
    \end{equation}
    where $\kappa := 2(\alpha_{\rho} - \zeta) \alpha_J$ and $\vect{b}(t) := \frac{1}{4\zeta} (\norm{\vect{\Lambda}_{\vect{x}_{\phi^c_l}}}\norm{\vect{f}_{\phi^c_l}} + \norm{\vect{\Lambda}_{\vect{x}_{\phi^c_l}}}\norm{\vect{w}_{\phi^c_l}} + \norm{\dot{\vect{\Gamma}}_{\phi^c_l}  \vect{e}^{\phi^c_l}}+ \norm{\vect{\Lambda}_{\Hat{\vect{x}}_{\phi^c_l}}} \norm{\dot{\Hat{\vect{x}}}_{\phi^c_l}})^2$. 
    
    Since \eqref{eq: Lyapunov function final upper bound, observation case} resembles \eqref{eq: Lyapunov function final upper bound}, to prove the satisfaction of the cluster's task, it suffices to prove the existence of an upper bound $\Bar{\vect{b}}(t)$ on $\vect{b}(t)$ as in Step~B. Consider the set $\mathcal{X}_{l_i}(t):=\{\hat{\vect{x}}^\top_{\phi_{l_i}}\in \mathbb{R}^{N^{TC}_{i}}| -1 < e^{\psi_{l_i}} < 0 \}$ containing the states $\Bar{\vect{x}}_{\phi_{l_i}}$ and the state estimates $\hat{\vect{x}}^{l_i}_{\phi_{l_i}}$ satisfying task $\phi_{l_i}$. Following Lemma~\ref{lemma on prescribed performance function}, suppose $\Gamma_{\phi_{l_i}}(t)$ assumes its maximum value at $t_{\text{max}}$. Then, $\mathcal{X}_{l_i}(t) \subseteq \mathcal{X}_{l_i}(t_{\text{max}})$ for all $t$, and $\mathcal{X}_{l_i}(t_{\text{max}})$ collects all states $\Hat{\vect{x}}_{\phi_{l_i}}$ such that $e^{\psi_{l_i}} \in (-1,0)$ at all $t \in \mathbb{R}_{\geq 0}$. 
    
    Since some robustness functions $\rho^{\psi_{l_i}}$ associated with tasks assigned to agent $l_i \in \mathcal{V}_l$ may depend on estimated states, i.e., $\rho^{\psi_{l_i}} = \rho^{\psi_{l_i}}(\hat{\vect{x}}_{\phi_{l_i}})$, condition (iii) in Assumption~\ref{Assumption on concavity and weelposeness of the robustness function} cannot be ensured by merely augmenting $\psi_{l_i}$ with the boundedness constraint $\psi^{\text{Bound}}_{l_i} := (\norm{\vect{x}_{\phi_{l_i}}}^2 \leq \Bar{C}_{l_i}^2)$. However, since the out-neighboring clusters of $\mathcal{C}_l$, i.e., $\mathcal{C}_j$ with $(\mathcal{C}_l, \mathcal{C}_j) \in \mathcal{E}'$, complete their tasks as established in Step~B, augmenting the original specifications $\psi_{l_i}$ and $\psi_{j_q}$ with the constraints $\psi^{\text{Bound}}_{l_i} := (\norm{\Bar{\vect{x}}_{\phi_{l_i}}}^2 \leq \Bar{C}_{l_i}^2)$ and $\psi^{\text{Bound}}_{j_q} := (\norm{\Bar{\vect{x}}_{\phi_{j_q}}}^2 \leq \Bar{C}_{j_q}^2)$, for all $l_i \in \mathcal{V}_l$ and $j_q \in \mathcal{V}_j$ with $(\mathcal{C}_l, \mathcal{C}_j) \in \mathcal{E}'$, together with the use of the convergent observer \eqref{Eq: State observer dynamics}, ensures that the augmented robustness functions associated to $\psi_{l_i} \land \psi^{\text{Bound}}_{l_i}$ and $\psi_{j_q} \land \psi^{\text{Bound}}_{j_q}$ are \mbox{well-posed}. To preserve the feasibility of $\psi_{l_i}$ and $\psi_{j_q}$ when they are extended to $\psi_{l_i} \land \psi^{\text{Bound}}_{l_i}$ and $\psi_{j_q} \land \psi^{\text{Bound}}_{j_q}$, the constants $\Bar{C}_{l_i}$ and $\Bar{C}_{j_q} \in \mathbb{R}_{\geq 0}$ must be selected sufficiently large so that conditions (i)-(iii) in Section~\ref{Section: Performance functions parameter design} remain satisfied for the augmented specifications. Intuitively, they should be chosen to account for the \mbox{worst-case} estimation errors induced by the observer, ensuring that the feasible set of the augmented tasks is nonempty.
    
    Under these conditions, $\mathcal{X}_{l_i}(t)$ remains bounded in time and $\mathcal{X}_{l_i}(t_{\text{max}})$ is guaranteed to be a bounded open set.
    As a result, for the same reasons as in Step~A, $\norm{\vect{f}_{\phi^c_l}}$, $\norm{\vect{w}_{\phi^c_l}}$, $ \norm{\dot{\vect{\Gamma}}_{\phi^c_l}  \vect{e}^{\phi^c_l}}$, $\norm{\vect{\Lambda}_{\vect{x}_{\phi^c_l}}}$ and $\norm{\vect{\Lambda}_{\Hat{\vect{x}}_{\phi^c_l}}}$ are bounded on $\mathcal{X}_{l_i}(t_{\text{max}})$. $\dot{\Hat{\vect{x}}}_{\phi^c_l}$ is a vector containing the dynamics of all estimates involved in tasks $\phi_{l_i}$, i.e., $\dot{\hat{x}}_{r}^{l_i} = - \rho_{r}^{l_i}(t)^{-1} J(e_{r}^{l_i}) \epsilon_{r}^{^{l_i}}(t)$, for all $l_i \in \mathcal{V}_l$ and $r \in  \Neigh{{l_i}}{T} \setminus \Neighext{{l_i}}{C}$. From the observer guarantees, $|\tilde{x}_{r}^{l_i}| < \delta_{r}^{l_i}$ holds. Thus, $\epsilon_{r}^{l_i} < \infty$, $J(e_{r}^{l_i}) < \infty$, and each component $\dot{\hat{x}}_{r}^{l_i}$ of $\dot{\Hat{\vect{x}}}_{\phi^c_l}$ is bounded. Hence, $\norm{\dot{\Hat{\vect{x}}}_{\phi^c_l}}$ is bounded as well.  Given the previous bounds, an upper bound ${\Bar{\vect{b}}}(t)$ on $\vect{b}(t)$ is guaranteed to exist for all $\hat{\vect{x}}_{\phi_{l_i}} \in \mathcal{X}_{l_i}(t_{\text{max}})$ and all $l_i \in \mathcal{V}_l$. 
    By introducing $S(\vect{e}^{\phi^c_l})$, task satisfaction can be proven as in Step~B. For this reason, the rest of the proof is omitted for brevity.
    
    In Step C, we analyzed clusters $\mathcal{C}_l$ whose out-neighbors are leaves. Owing to the acyclic structure of $\mathcal{G}'$, and assuming task satisfaction for the already treated clusters, the same reasoning can be applied iteratively to ensure task satisfaction all the way up to the root clusters.
\end{proof}
\begin{remark}
    As introduced in Section~\ref{Section: cluster induced graph}, the global task $\phi$ can be rewritten as $\phi = \land_{l=1}^{N'} {\phi}^c_l$. Thus, since Theorem~\ref{Theorem on task satisfacion} guarantees the satisfaction of each cluster’s task, global task satisfaction follows.
\end{remark}

%% file: Content/Case_study.tex
\begin{figure*}[t!]
    \centering
    \includegraphics[width=1\linewidth]{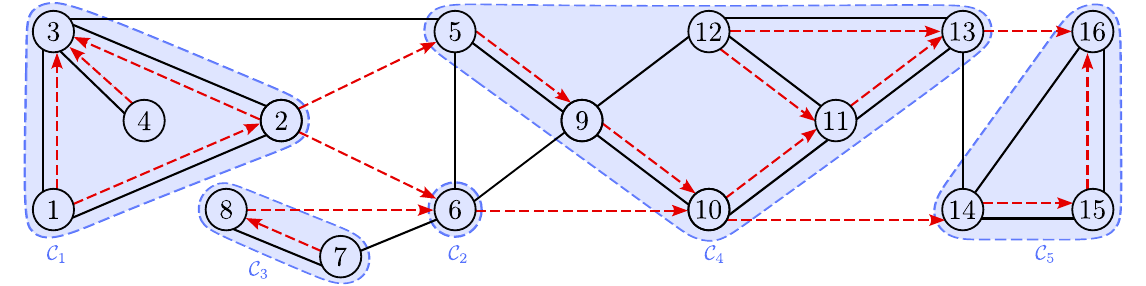}
    \caption{The graphs $\mathcal{G}_C$ and $\mathcal{G}_T$, respectively in black and red, induce five clusters $\mathcal{C}_1$, $\mathcal{C}_2$, $\mathcal{C}_3$, $\mathcal{C}_4$, $\mathcal{C}_5$ with $\mathcal{V}_1 :=\{1,2,3,4\}$, $\mathcal{V}_2 := \{6\}$, $\mathcal{V}_3 := \{7,8\}$, $\mathcal{V}_4:= \{5,9,10,11,12,13\}$ and $\mathcal{V}_5 := \{14,15,16\}$.}
    \label{fig:Graph Gt, Gc, induced clusters simu}
    \vspace{-0.1cm}
\end{figure*}

To demonstrate the validity of the proposed approach, we consider a network of $N = 16$ mobile robots whose information exchange is described by the communication graph $\mathcal{G}_C$ shown in black in Fig.~\ref{fig:Graph Gt, Gc, induced clusters simu}. Each robot evolves according to the dynamics in \cite{liu2008omni}, i.e.:
\begin{equation}
    \dot x_i(t) = - f_i(\vect{x}) + A_i(B_i^\top)^{-1} R_i u_i + w_i(t), 
\end{equation}
where: $x_i(t) := [x_{i,1}, x_{i,2}, x_{i,3}]^\top$ denotes the state of robot $i$, consisting of its planar position $\vect{p}_i := [x_{i,1}, x_{i,2}]^\top$ and orientation $x_{i,3}$; $u_i := [\omega_{i,1}, \omega_{i,2}, \omega_{i,3}]$ consists of the wheel angular velocities; $R_i$ denotes the wheel radius; $A_i := \begin{bmatrix} \cos(x_{i,3}) & -\sin(x_{i,3}) & 0 \\ \sin(x_{i,3}) & \cos(x_{i,3}) & 0 \\ 0 & 0 & 1\end{bmatrix}$ is the rotation matrix associated with the robot orientation; $B_i := \begin{bmatrix} 0 & \cos(\pi/6) & -\cos(\pi/6) \\ -1 & \sin(\pi/6) & \sin(\pi/6) \\ L_i & L_i & L_i \end{bmatrix}$, with $L_i$ denoting the radius of the robot body; $f_i(\vect{x}) := \sum_{j \in \mathcal{V}\setminus\{i\}} k_i \frac{[(\vect{p}_i- \vect{p}_j)^\top, 0]\top}{\norm{\vect{p}_i-\vect{p}_j}+ \epsilon}$ represents a bounded coupling term promoting collision avoidance among the robots, where $k_i = 0.1$ is a design parameter and $\epsilon = 0.0001$ prevents singularities; and $w_i(t) \in [-1.5, 1.5]$ denotes a bounded random disturbance affecting robot $i$.

The robots are subject to the following STL tasks: 
\begin{equation*}
    \begin{split}
         \phi_1 := G_{[4,5]}&((\norm{\vect{p}_1 - [0, 2]^\top}^2 \leq 7.05) \land (\norm{\vect{p}_1 - \vect{p}_2}^2 \leq\\ & 26.75) \land (\norm{\vect{p}_1 - \vect{p}_3}^2 \leq 70.05))\\
         \phi_2 := G_{[4,5]}&( (\norm{\vect{p}_2 - \vect{p}_3}^2 \leq 70.05) \land  (\norm{\vect{p}_2 - \vect{p}_5}^2 \leq \\ &70.05)  \land  (\norm{\vect{p}_2 - \vect{p}_6}^2 \leq 70.05)) \\
         \phi_3 := G_{[4,5]} &(\norm{\vect{p}_3 - [-1.175, -1.618]^\top}^2 \leq 7.05) \\
         \phi_4 := G_{[4,5]} &(\norm{\vect{p}_4 - \vect{p}_3}^2 \leq 26.75)\\
         \phi_5:= G_{[4,5]} & (\norm{\vect{p}_5 - \vect{p}_9}^2 \leq 70.04) \\
         \phi_6:= G_{[4,5]} & (\norm{\vect{p}_6 - \vect{p}_{10}}^2 \leq 70.04)\\
         \phi_7:= G_{[4,5]} & (\norm{\vect{p}_7 - \vect{p}_8}^2 \leq 70.04)\end{split}
\end{equation*}
\begin{equation*}
    \begin{split}
         \phi_8:= G_{[4,5]} & (\norm{\vect{p}_8 - \vect{p}_6}^2 \leq 70.04)\\
         \phi_9:= G_{[4,5]} & (\norm{\vect{p}_9 - \vect{p}_{10}}^2 \leq 70.04)\\
         \phi_{10}:= G_{[4,5]} & ((\norm{\vect{p}_{10} - \vect{p}_{11}}^2 \leq 70.04) \land (\norm{\vect{p}_{10}- \vect{p}_{14}}^2 \leq\\ & 70.04)) \\ 
         \phi_{11}:= G_{[4,5]} & (\norm{\vect{p}_{11} - \vect{p}_{13}}^2 \leq 70.04) \\        
         \phi_{12}:= G_{[4,5]} & ((\norm{\vect{p}_{12} - \vect{p}_{11}}^2 \leq 70.04) \land (\norm{\vect{p}_{12}- \vect{p}_{13}}^2 \leq\\ & 70.04)) \\
         \phi_{13}:= G_{[4,5]} & (\norm{\vect{p}_{13} - \vect{p}_{16}}^2 \leq 70.04) \\
         \phi_{14}:= G_{[4,5]} & (\norm{\vect{p}_{14} - \vect{p}_{15}}^2 \leq 70.04) \\
         \phi_{15}:= G_{[4,5]} & (\norm{\vect{p}_{15} - \vect{p}_{16}}^2 \leq 70.04) \\
        \phi_{16} := G_{[4,5]} &(\norm{\vect{p}_{16} - [-1.175, -1.618]^\top}^2 \leq 7.05)    
    \end{split}
\end{equation*}

Intuitively, Agents $1,3$ and $16$ are required to reach and remain within a prescribed spatial region, while the remaining agents are required to follow the swarm and maintain sufficient proximity to each other.

The task graph $\mathcal{G}_T$ resulting from the STL task is depicted in red in Fig.~\ref{fig:Graph Gt, Gc, induced clusters simu}. The corresponding \mbox{cluster-induced} graph $\mathcal{G}' = (\mathcal{C}', \mathcal{E}')$, obtained from $\mathcal{G}_C$ and $\mathcal{G}_T$, is defined by the cluster set $\mathcal{C}' := \{\mathcal{C}_1, \mathcal{C}_2, \mathcal{C}_3, \mathcal{C}_4, \mathcal{C}_5\}$ and the edge set $\mathcal{E}' := \{(\mathcal{C}_1, \mathcal{C}_2), (\mathcal{C}_1, \mathcal{C}_4), (\mathcal{C}_3, \mathcal{C}_2), (\mathcal{C}_2, \mathcal{C}_4), (\mathcal{C}_4, \mathcal{C}_5)\}$, where $\mathcal{V}_1 :=\{1,2,3,4\}$, $\mathcal{V}_2 := \{6\}$, $\mathcal{V}_3 := \{7,8\}$, $\mathcal{V}_4:= \{5,9,10,11,12,13\}$ and $\mathcal{V}_5 := \{14,15,16\}$. Since $\mathcal{G}_T$ satisfies Assumption~\ref{Assumption on Task graph}, $\mathcal{G}'$ is a directed acyclic graph as introduced in Remark~\ref{remark on the acyclicity of the cluster induced graph}.

To allow every $i \in \mathcal{V}$ to estimate the state of the agents $N_j^i \in \Nhop{i}{k} \cap \Neigh{i}{T}$, the \mbox{$k$-hop} PPSO in \eqref{Eq: State observer dynamics} is implemented with $k = 3$. Then, $\Nhop{1}{3} :=\{4,5,6,9\}$, $\Nhop{2}{3} :=\{4,5,6,9\}$, $\Nhop{3}{3} :=\{6,7,9,10,12\}$, $\Nhop{4}{3} :=\{1,2,5,6,9\}$, $\Nhop{5}{3} :=\{1,2,4,7,8,10,11,12,13\}$, $\Nhop{6}{3} :=\{1,2,3,4,8,10,11,12,13\}$, $\Nhop{7}{3} :=\{3,5,9,10,12\}$, $\Nhop{8}{3} :=\{5,6,9\}$, $\Nhop{9}{3} :=\{1,2,3,4,7,8,11,13,14\}$, $\Nhop{{10}}{3} :=\{3,5,6,7,12,13,14\}$, $\Nhop{{11}}{3} :=\{5,6,9,14,15,16\}$, $\Nhop{{12}}{3} :=\{3,5,6,7,10,14,15,16\}$, $\Nhop{{13}}{3} :=\{5,6,9,10,15,16\}$, $\Nhop{{14}}{3} :=\{9,10,11,12\}$, $\Nhop{{15}}{3} :=\{11,12,13\}$ and $\Nhop{{16}}{3} :=\{11,12,13\}$. 

By exploiting the properties in \eqref{eq: always operator properties}, each task $\phi_i$ can be rewritten as $\phi_i = \land_{j = 1}^{z_i} \phi_{i,j}$, with $\phi_{i,j} := G_{[4,5]} (\norm{\vect{p}_i - \vect{p}_l} \leq r_{\psi_{i,j}})$, and where $l$ denote the global index of the agent involved, together with agent $i$, in the satisfaction of task $\phi_{i,j}$. Thus, according to Section~\ref{Section: Performance functions parameter design}, the task funnel relaxation can be computed, for all $\phi_{i,j}$ involving state estimates, as in Example~\ref{example: funnel modification}, and task satisfaction can be enforced following Remark~\ref{Remark: task modification in case of composite tasks}.
To guarantee feasibility of the STL tasks, the prescribed performance functions of the $k$-hop PPSO, i.e., $\delta^{N_j^i}_i(t)$ for all $i\in \mathcal{V}$ and $N_i^j \in \Nhop{i}{3}$, are tuned such that conditions (i)-(iii) in Section~\ref{Section: Performance functions parameter design} are satisfied with $\rho_{\psi_{i,j}}^t(t) = 2 \norm{\delta_{\psi_{i,j}}^i(t)} r_{\psi_{i,j}} -  \norm{\delta_{\psi_{i,j}}^i(t)}^2$, where $\norm{\delta_{\psi_{i,j}}^i(t)}$ is the norm of the prescribed performance functions associated with the state estimates involved in $\phi_{i,j}$. For example, task $\phi_2$ involves the estimates $\estimate{\vect{p}}{2}{5}:= [\estimate{x}{2}{{5,1}}, \estimate{x}{2}{{5,1}}]^\top$ and $\estimate{\vect{p}}{2}{6}:= [\estimate{x}{2}{{6,1}}, \estimate{x}{2}{{6,1}}]^\top$. Thus, to guarantee feasibility, the prescribed performance functions of the state estimation errors are selected as $\delta^2_{5, j}(t) \simeq 5.475 e^{-0.5t} + 0.209$ and $\delta^2_{6,j}(t) \simeq  14.613e^{-0.5t} + 0.557$ for all $j\in \{1,2,3\}$. Then, to guarantee the prescribed performance of the observer, $\rho^{N_r^5}_{5,j}(t)$ and $\rho^{N_r^6}_{6,j}(t)$, for all  $N_r^5 \in \Nhop{5}{3}$, $N_r^6 \in \Nhop{6}{3}$ and $j\in \{1,2,3\}$, are designed to satisfy the constraint in \eqref{eq: constraint on the prescribed prformance functions}. In particular, $\rho^2_{5,j}(t) \simeq 0.1 e^{-0.5t} + 0.004$ and $\rho^2_{6,j}(t) \simeq 0.735 e^{-0.5t} + 0.028$ for all $j\in \{1,2,3\}$.
In order to satisfy the initialization condition in Assumption~\ref{Assumption on initialization}, the state estimation error are randomly initialized such that $|\xi^{N^i_r}_{{i,j}}(0)| < \rho^{N_r^i}_{i,j}(0)$ holds for all $i\in \mathcal{V}$, $N_r^i \in \Nhop{i}{k}$ and $j\in \{1,2,3\}$.

Since $\mathcal{G}_T$ satisfies Assumption~\ref{Assumption on Task graph},  Theorem~\ref{Theorem on task satisfacion} guarantees that the control law in \eqref{eq: agent's input} ensures task satisfaction of the MAS. Specifically, by following the arguments in the proof of Theorem~\ref{Theorem on task satisfacion}, task satisfaction can first be established for the leaf cluster $\mathcal{C}_5$. Then, owing to the boundedness of the control inputs $u_i(t)$, for all $i \in \mathcal{C}_5$, and the resulting prescribed performance guarantees provided by the $k$-hop PPSO, namely $|\error{x}{N_j^i}{i}(t)| <\delta^{N_j^i}_i(t)$ for all $i\in \mathcal{C}_5$ and $N_j^i \in \Nhop{i}{k}$, task satisfaction for cluster $\mathcal{C}_4$ follows. Proceeding recursively along the task-graph hierarchy, task satisfaction can subsequently be established for cluster $\mathcal{C}_2$, and finally for clusters $\mathcal{C}_1$ and $\mathcal{C}_3$. This recursive implication follows directly from the acyclic dependency structure of the task graph.

\begin{figure}
    \centering
    \includegraphics[width =\linewidth]{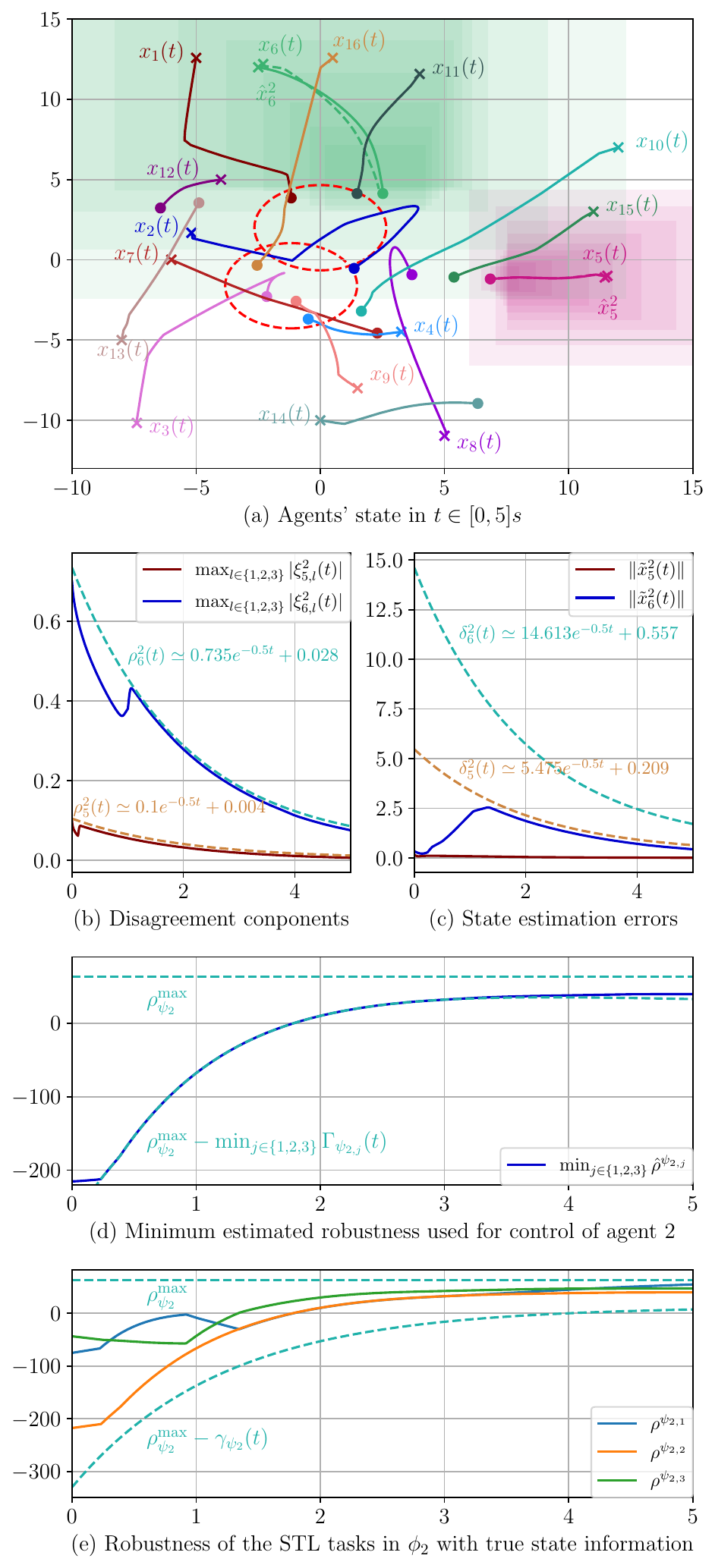}
    \caption{(a) State and estimated trajectories; initial states are represented by crosses, terminal state by dots. The pink and green regions represent the evolution over time of the upper bounds on the state estimation errors guaranteed by the $k$-hop PPSO; red circles denote the target regions of the individual tasks. (b) Evolution of the disagreement components on the estimates performed by agent $2$; the dashed lines are the prescribed performance functions of the $k$-hop PPSO. (c) Evolution of the norm of the state estimation errors performed by agent $2$. (d) Minimum estimated robustness of $\phi_2$. (e) Robustness of the tasks in $\phi_2$ evaluated using the true state.}
    \label{fig: Simulation results}
\end{figure}

Numerical simulations were conducted in Python on a computer equipped with an  Intel Core Ultra $7$ $165\mathrm{H}$ processor ($3.8 \mathrm{GHz}$). The results are shown in Fig.~\ref{fig: Simulation results}. Specifically, Fig.~\ref{fig: Simulation results}a illustrates the MAS trajectory together with the evolution of the state estimates required for task $\phi_2$, as well as the associated uncertainty regions determined by the upper bounds on the estimation errors guaranteed by the $k$-hop PPSO. Due to space limitation, we present here the dynamics of the state estimation errors and the robustness functions for agent~$2$ only, as shown in Figs.~\ref{fig: Simulation results}b-e. 
As illustrated in Fig.~\ref{fig: Simulation results}b and Fig.~\ref{fig: Simulation results}c, whenever the disagreement vector satisfies the prescribed performances bounds in \eqref{eq: definition of the bounds on the disagreement term}, the estimation error evolves according to the inequality \eqref{eq: state estimation error inequality}. Moreover, Fig.~\ref{fig: Simulation results}d and Fig.~\ref{fig: Simulation results}e show that enforcing $\min_{j \in \{1,2,3\}} \hat{\rho}^{\psi_{2,j}} > 0$ for all $t \in [0,5]$, through the constraint in \eqref{eq: inequality on robustness to guarantee task satisfaction, observer not in design stage of gamma}, with $\Gamma_{\psi_2}(t) =\min_{j \in \{1,2,3\}} \Gamma_{\psi_{2,j}}$, is sufficient to guarantee satisfaction of $\phi_2$. Indeed, Fig.~\ref{fig: Simulation results}e confirms that $\rho^{\psi_{2,j}} > 0$ for all $j \in \{1,2,3\}$ and $t\in [4,5]$.

As shown in Fig.~\ref{fig: Simulation results}c, the norm of the state estimation error $\error{x}{2}{6}$ reaches a peak of approximately $2.5 \mathrm{m}$ during the transient phase and remains strictly positive for all $t \in [4,5]$. Together with the robustness plots in Fig.~\ref{fig: Simulation results}e, these results demonstrate that the proposed framework guarantees task satisfaction for all estimation errors satisfying the prescribed performance bounds of the observers, while preserving feasibility of the task during the transient phase. This confirms the effectiveness of the proposed framework.

%% file: Content/Conclusions.tex
We proposed a decentralized \mbox{observer-based} control framework for MAS subject to STL specifications under communication-task graphs mismatch. To enable each agent to estimate the states of non-neighboring agents involved in its tasks, we developed a $k$-hop Prescribed Performance State Observer that provides state estimates of agents up to $k$ communication hops away while guaranteeing predefined transient and steady-state performance.
Owing to its decentralized modular architecture, the proposed framework is robust to bounded external disturbances and is well suited for \mbox{large-scale} heterogeneous \mbox{multi-agent} systems. Future work will focus on \mbox{observer–controller} \mbox{co-design} strategies aimed at relaxing Assumption~\ref{Assumption on the non-existance of a task not in communication inside a cluster}, as well as further enhancing scalability and performance guarantees.